\UseRawInputEncoding

\documentclass[10pt,letterpaper]{article}
\usepackage[top=0.85in,left=2.75in,footskip=0.75in]{geometry}

\usepackage{amsmath,amssymb}

\usepackage{changepage}

\usepackage{textcomp,marvosym}

\usepackage{cite}

\usepackage{microtype}
\DisableLigatures[f]{encoding = *, family = * }

\usepackage[table]{xcolor}

\usepackage{array}

\usepackage{amsfonts}
\usepackage{amsthm}
\usepackage{stmaryrd}
\usepackage{mathrsfs}
\usepackage{float}
\usepackage{psfrag}
\usepackage{multirow}
\usepackage{multicol}
\usepackage{lipsum}
\usepackage{mwe}
\usepackage{epsfig, epic, eepic, epsf}
\usepackage{chemarr}
\usepackage{subcaption}
\usepackage[overload]{empheq}
\usepackage{threeparttable}
\usepackage{changes} 
\usepackage{overpic}
\usepackage{siunitx}
 \usepackage{pdfpages}
 \usepackage{appendix}
 \usepackage{siunitx}

\newcolumntype{+}{!{\vrule width 2pt}}

\newlength\savedwidth

\newcommand\thickhline{\noalign{\global\savedwidth\arrayrulewidth\global\arrayrulewidth 2pt}%
\hline
\noalign{\global\arrayrulewidth\savedwidth}}

\raggedright
\usepackage[aboveskip=1pt,labelfont=bf,labelsep=period,justification=raggedright,singlelinecheck=off]{caption}

\newtheorem{thm}{Theorem}[section]

\newtheorem{lem}[thm]{Lemma}
\newtheorem{prop}[thm]{Proposition}

\usepackage[utf8]{inputenc}
\usepackage{algorithm}
\usepackage{algorithmic}
\usepackage{verbatim}

\makeatletter
\renewcommand{\@biblabel}[1]{\quad#1.}
\newenvironment{breakablealgorithm}
{
	\begin{center}
		\refstepcounter{algorithm}
		\hrule height.8pt depth0pt \kern2pt
		\renewcommand{\caption}[2][\relax]{
			{\raggedright\textbf{\ALG@name~\thealgorithm} ##2\par}%
			\ifx\relax##1\relax 
			\addcontentsline{loa}{algorithm}{\protect\numberline{\thealgorithm}##2}%
			\else 
			\addcontentsline{loa}{algorithm}{\protect\numberline{\thealgorithm}##1}%
			\fi
			\kern2pt\hrule\kern2pt
		}
	}{
		\kern2pt\hrule\relax
	\end{center}
}
\makeatother

\usepackage{lastpage,fancyhdr,graphicx}
\usepackage{epstopdf}
\fancyheadoffset[L]{2.25in}
\fancyfootoffset[L]{2.25in}
\begin{document}
\vspace*{0.2in}

\begin{flushleft}
{\Large
\textbf\newline{Patch formation driven by stochastic effects of interaction between viruses and defective interfering particles} 
}
\newline
\\
Qiantong Liang\textsuperscript{1},
Johnny Yang\textsuperscript{2},
Wai-Tong Louis Fan\textsuperscript{2},
Wing-Cheong Lo\textsuperscript{1*}
\\
\bigskip
\textbf{1} Department of Mathematics, City University of Hong Kong, Hong Kong, Hong Kong, China
\\
\textbf{2} Department of Mathematics, Indiana University, Bloomington, IN, United States
\\
\bigskip

%
%





* wingclo@cityu.edu.hk

\end{flushleft}
\section*{Abstract}
Defective interfering particles (DIPs) are virus-like particles that occur naturally during virus infections. These particles are defective, lacking essential genetic materials for replication, but they can interact with the wild-type virus and potentially be used as therapeutic agents.
However, the effect of DIPs on infection spread is still unclear due to complicated stochastic effects and nonlinear spatial dynamics. In this work, we develop a model with a new hybrid method to study the spatial-temporal dynamics of viruses and DIPs co-infections within hosts. We present two different scenarios of virus production and compare the results from deterministic and stochastic models to demonstrate how the stochastic effect is involved in the spatial dynamics of virus transmission.
We quantitatively study the spread features of the virus, including the formation and the speed of virus spread and the emergence of stochastic patchy patterns of virus distribution. Our simulations simultaneously capture observed spatial spread features in the experimental data, including the spread rate of the virus and its patchiness. The results demonstrate that DIPs can slow down the growth of virus particles and make the spread of the virus more patchy.

\section*{Author summary}
Defective interfering particles (DIPs) are viral mutants in which a crucial part of the particle's genome has been lost. DIPs are not infectious but can still co-infect cells with natural viruses. Such mutations are not uncommon. In fact, it has been found in most classes of viruses, including SARS coronavirus and influenza virus. It gives DIPs a promising future as a medium for disease treatment. However, the mechanism by which DIPs affect virus transmission remains unclear. In this paper, we develop a model to study the interaction between viruses and DIPs within host cells and the role stochastic effects play in virus transmission. Our simulations can capture patchy patterns and other spatial spread features observed in experiments and demonstrate that DIPs can slow down the growth of virus particles and make the spread of viruses more patchy.

\section*{Introduction}
Many diseases such as COVID-19, Ebola virus disease, AIDS, and SARS, are caused by the transmission of viruses. 
Various antiviral drugs have been proposed to inhibit the gene and protein functions of viruses. Still, a major challenge in drug development is caused by occasional mutations in the viral genomes. However, some of these mutations
may help us create a new type of treatment through developing defective interfering particles (DIPs), 
which are virus-like particles that have been detected in patients infected with influenza A virus~\cite{Frensing2015}, and with dengue virus, as well as birds infected with West Nile virus~\cite{Pesko2012}. DIPs lack some viral genes that are essential for replication. But, when they co-infect a cell with viable viruses, DIPs divert replication or packaging resources from the virus towards their own growth, thereby compromising normal virus growth~\cite{Baltes2017, Frensing2015}. 
The competition between infectious viruses and DIPs for the resources in a host may induce a delay and decrease in infectious virus production~\cite{akpinar2016high, Baltes2017}. For example, 
in the recent work~\cite{ rezelj2021defective}, a combined experimental evolution and computational approach identified
 defective viral genomes that optimally interfere with Zika virus infection and show antiviral activity in mice and mosquitoes.
Therefore, DIPs interfere with virus production, a feature that underscores their promise as therapeutic agents~\cite{dropulic1996conditionally, noble2004interfering, li2021dengue, rezelj2021defective}. 

In a recent experimental study~\cite{Baltes2017}, engineered reporter viruses and DIP were constructed, which enabled measurement of the gene expressions of both viral and DIP during co-infection of susceptible host cells. Quantitative microscopy imaging in~\cite{Baltes2017}
demonstrated that levels of virus and DIP production from co-infected cells can be highly sensitive to their input ratios (multiplicities of infection, MOI), and revealed diverse spatial patterns during co-infection spread.
The experimental results showed that viral gene expression was more delayed and that patterns of spread became more ``patchy" with a higher level of DIPs to the initial cell. However, it is not clear that how the timing and level of this spatial distribution of DIP expression are related to the spread of virus infection, and what are the key mechanisms responsible for the diverse spatial patterns of the virus and DIP levels.

Many mathematical models were built to study the growth of virus~\cite{Wodarz2012, Saenz2010, Getto2008, Perelson2002, Heldt2015, Perelson2013, Pawelek2012, Graw2016, Whitman2020, Yin2018} and the interaction of DIPs and viruses~\cite{Kirkwood1994, Frank2000, Akpinar2016a, Laske2016, Mapder2019, Saxena2018}. The simulations and analyses provide us a theoretical idea to understand the development of infectious diseases and how to control the growth of viruses. 
For example, in~\cite{Kirkwood1994}, a simple mathematical model was proposed for studying the deterministic chaos caused by DIPs. However, there are not many models considering the spatial effect of the interaction of DIPs and viruses in a one- or two-dimensional domain. 
Frank~\cite{Frank2000} developed a one-dimensional partial differential equation model for studying the dynamics of the populations of DIPs and viruses within hosts. 
His work studied how the dynamics of virus spreading depend on the rate at which killed host cells are replaced. These results explain the key processes that control the diversity of observed experimental outcomes and provide a stepping stone to study the spatial model of the transmissions of DIPs and viruses. A two-dimensional domain has to be considered for reproducing the patchy pattern. Akpinar et al.~\cite{Akpinar2016a} built a two-dimensional computational model, adapting a cellular automaton approach to incorporate kinetic data on virus growth, but the model is not able to capture the spread rate and the spatial patterns simultaneously observed in~\cite{Baltes2017}.

The existing computational studies provide a keystone for modeling the interaction of DIPs and wild-type viruses. However, the mechanism by which DIPs affect the spatial distribution of virus expression is still unclear partly due to complicated stochastic effects and nonlinear spatial dynamics. In~\cite{Pearson2011}, the authors applied a stochastic model to study different solutions for continuous and burst production of virions which cannot be studied through deterministic models. In~\cite{Immonen2012}, a hybrid stochastic-deterministic computational model was applied to capture experimentally observed variation in the fitness difference between two virus strains. The simulations of the model suggest a way to minimize the variation and dual infection in experiments. In~\cite{Clark2011}, a stochastic model was built to study the effect of DIPs and the results support that DIPs have a slowing effect on the growth of viral plaques, but the spread features are not quantified in that study.
These computational studies suggest that stochastic effects play an important role in virus spreading, but the stochastic effects in the virus and DIP transmissions are poorly understood. It inspires us to build a stochastic spatial model to study the interaction of DIPs and viruses and how the effect of DIPs leads to patchy patterns of virus expression observed in experiments.

In this paper, we develop and analyze a new mathematical model to study the spreading speed and the spatial pattern generated by the interaction of viruses and DIPs. To incorporate the random movements of the virus and the DIPs and the stochastic effect of the interactions due to finite number of particles, we developed a stochastic reaction-diffusion system for the virus and DIP co-infection and built a hybrid method for stochastic simulation. Our stochastic model enables also the study and comparison of two common scenarios of virus production. Our simulation results demonstrated that this model can regenerate simultaneously the patchy patterns and the spread rates observed in wet-lab experiments~\cite{Baltes2017}, which was not achieved in previous studies~\cite{Akpinar2016a}.

\section*{Modeling}
Our new hybrid model is developed based on the deterministic reaction-diffusion model introduced by Frank~\cite{Frank2000}, but has several differences and new features. Importantly, our model and simulation results capture spatial spread features in two-dimensions observed in experiments and overcome computational challenges in stochastic simulations in two-dimensional domains, while the results in~\cite{Frank2000} are for one-dimension. Furthermore, we introduced and compared two different scenarios of virus production in the stochastic simulations.

Below we describe firstly the deterministic part of our model which is a system of partial differential equations, and secondly our stochastic model that incorporates two different scenarios of virus production.

\subsection*{Deterministic model}

Based on the model in~\cite{Frank2000}, we propose a new model which includes the virus and DIP productions. As shown in Fig~\ref{system}, in the model, we consider free natural infectious virus, denoted by $V(t,\vec{x})$, and defective interfering particles (DIPs), denoted by $D(t,\vec{x})$ where $\vec{x}=(x_1,x_2)$ is a vector which represents a spatial location in a two-dimensional domain $[0, x_{1\max}]\times [0, x_{2\max}]$. Also, there are six types of cells: uninfected cells, cells infected only by natural viruses but not in the period of virus production, cells infected only by natural viruses and in the period of virus production, cells infected by DIPs only, cells infected by DIPs and natural viruses but not in the period of virus production, and cells infected by both DIPs and natural viruses as well as in the period of virus production. The numbers of the respective cells are denoted by $C$, $C_V$, $C_V^*$, $C_D$, $C_{VD}$ and $C^*_{VD}$, respectively. 

\begin{figure}[!ht]
\includegraphics[width=0.85\textwidth]{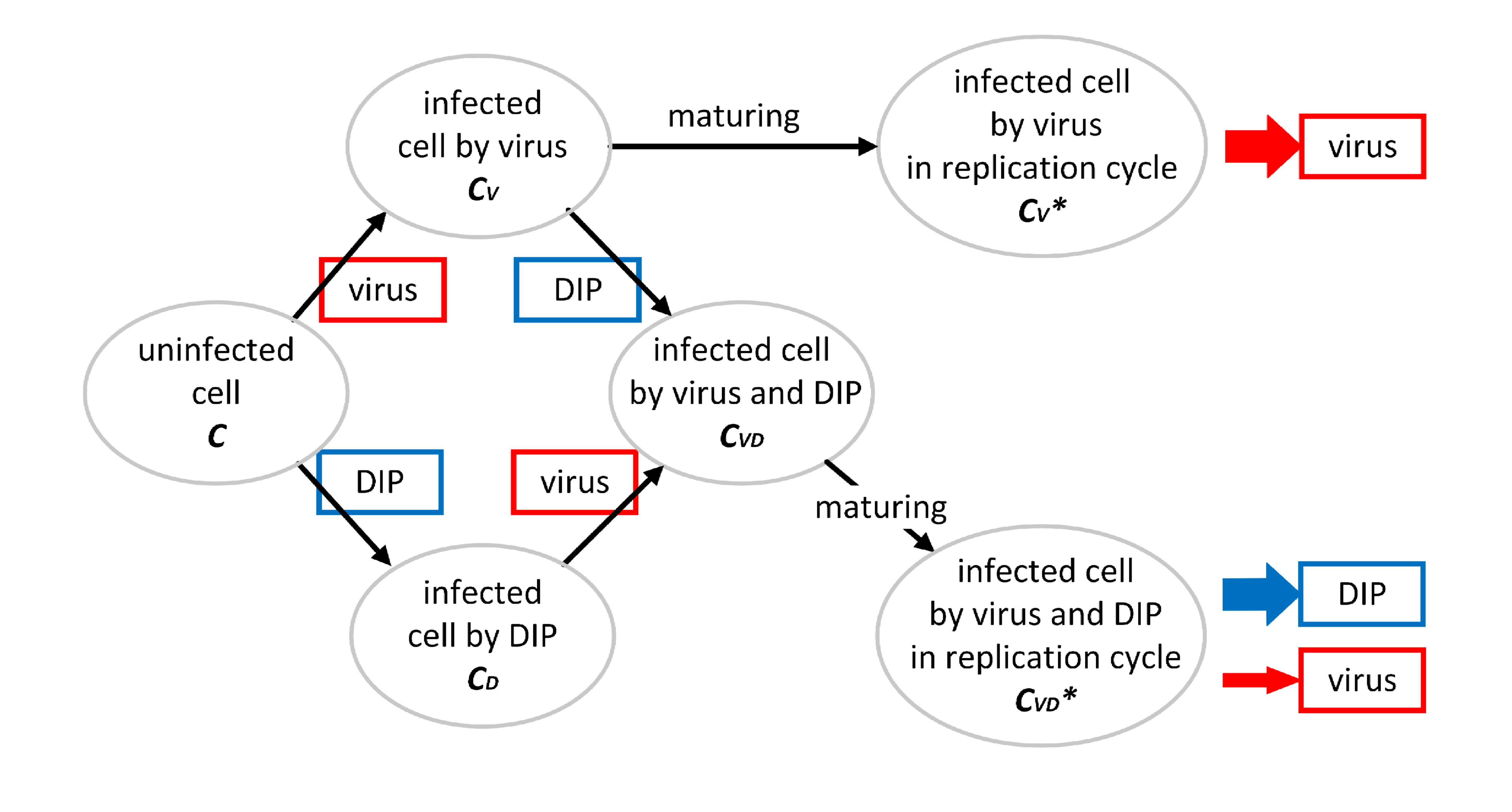}
\caption{{\bf Schematic diagram of the virus-DIPs system.}}
\label{system}
\end{figure}

Considering infection by DIPs of cells late in the replication cycle is too late to affect the production of the virus, we assume that DIPs cannot infect the cells $C^*_V$~\cite{Frank2000}.
There are two age categories for each of the $C_V$ and the $C_{VD}$ cells in our model. Ultimately, the DIPs are produced only by the mature $C^*_{VD}$ cells, and virus particles are produced by both $C^*_{V}$ and $C^*_{VD}$ cells. While $C_V$ cells can get to maturity by themselves and produce virus, $C_D$ cells \textit{cannot} produce DIP unless they are co-infected by the virus and become $C_{VD}$ and get to maturity. The latter models the situation that DIPs cannot replicate unless they co-infect a cell with a wild-type virus.

The following two equations are for modeling the dynamics of the virus and DIP:
\begin{eqnarray}
\begin{aligned}
\frac{\partial V}{\partial t} &= \underbrace{d_V \nabla^2 V}_{\text{Diffusion}} + \underbrace{\alpha_1 C^*_V + \alpha_2 C^*_{VD} }_{\text{virus production}} - \underbrace{\delta_V V}_{\text{clearance of virus}}, \\
\frac{\partial D}{\partial t} &= \underbrace{d_D \nabla^2 D}_{\text{Diffusion}} + \underbrace{\alpha_3 C^*_{VD}}_{\text{DIP production}} - \underbrace{\delta_D D}_{\text{clearance of DIP}},
\end{aligned}
\label{equation1}
\end{eqnarray}

\noindent
where $\nabla^2$ is the Laplacian operator, describing the virus and DIP diffusion.

We assume that cells are not moving in the spatial domain and we can model the dynamics of the cell densities by the following system:
\begin{eqnarray}
\begin{aligned}
    \frac{\partial C_V^*}{\partial t} &=\underbrace{\nu_1 C_V}_{\text{maturing for virus production}} - \underbrace{\beta_1 C_V^*}_{\text{cell death}},\\
    \frac{\partial C_{VD}^*}{\partial t}&=\underbrace{\nu_2 C_{VD}}_{\text{maturing for virus production}}-\underbrace{\beta_2C_{VD}^*}_{\text{cell death}},\\
\frac{\partial C}{\partial t} &=  \underbrace{\alpha_C C \left(1- C_T/K\right)}_{\text{Cell growth}}-  \underbrace{\gamma_1 CV}_{\text{infected by normal viruses}} - \underbrace{\gamma_2 CD}_{\text{infected by DIPs}} -\underbrace{\delta_C C}_{\text{cell death}},\\
\frac{\partial C_V}{\partial t} &= \underbrace{\gamma_1 CV}_{\text{infected by normal viruses}} - \underbrace{\gamma_2 C_V D}_{\text{infected by DIPs}}  - \underbrace{\nu_1C_V}_{\text{maturing for virus production}}  - \underbrace{\delta_{CV} C_V}_{\text{cell death}}, \\
\frac{\partial C_D}{\partial t} &=   \underbrace{\gamma_2 CD}_{\text{infected by DIPs}} - \underbrace{\gamma_1 C_D V}_{\text{infected by normal viruses}} - \underbrace{\delta_{CD} C_D}_{\text{cell death}},\\
\frac{\partial C_{VD}}{\partial t} &= \underbrace{ \gamma_2 C_V D}_{\text{infected by DIPs} }+ \underbrace{\gamma_1 C_D V}_{\text{infected by normal viruses}}  -  \underbrace{\nu_2C_{VD}}_{\text{maturing for virus production}} - \underbrace{\delta_{CVD} C_{VD}}_{\text{cell death}},
\end{aligned}
\label{equation1_1}
\end{eqnarray}
where $C_T = C+C_V+C_D+C_{VD}+C^*_V +C^*_{VD}$ is the total density of all cells.

Our model Eqs~(\ref{equation1})-Eqs~(\ref{equation1_1}) is different from that of~\cite{Frank2000} in several ways: (i) there are two age categories for $C_{VD}$ cells in our model, but there is no age structure for $C_{VD}$ cells in~\cite{Frank2000}; (ii) the mature $C^*_{VD}$ cells can produce virus in our model, but the $C_{VD}$ cells in~\cite{Frank2000} cannot produce virus; (iii) $C_D$ cells cannot recover to be uninfected cells in our model, but they can recover in~\cite{Frank2000}; (iv) our parameters $\gamma_1$ and $\gamma_2$ can be different, but they are the same in~\cite{Frank2000}.

\subsection*{Stochastic models in two different scenarios of virus production}
Since the outcome of the model with DIPs is sensitive to the competition between viruses and DIPs, different kinds of perturbation to the production of viruses and DIPs may contribute to a huge change in the probability distribution of the outcome. The study in~\cite{Pearson2011} suggested that there are two scenarios of virus production, which can create different kinds of perturbations to virus production:

\begin{enumerate}
    \item[] Scenario 1:$\;$ infected cells produce virus and DIPs through cell bursting;
    \item[] Scenario 2:$\;$ infected cells keep producing viruses and DIPs continuously.
\end{enumerate}

However, these two scenarios cannot be distinguished by our deterministic PDE model~\cite{Pearson2011} as both models with different scenarios have identical mean-field kinetics. In this study, we built a stochastic model and developed an efficient simulation method to examine the effects on the spatial distribution of viruses under different scenarios.

Due to the high computational cost of the spatial stochastic model, there are not many studies considering the effects of different scenarios for virus production on the spreading speed and distribution of the virus. To improve the computational efficiency, here we simulate our model with Spatial Stochastic Simulation Algorithm (SSA)~\cite{smith2018spatially}, which is a method to generate an exact sample from the probability mass function that is the solution of the chemical master equation.

In SSA, we consider the spatial domain as a two-dimensional square with length $L$. The domain is partitioned into $N_c \times N_c$ identical compartments that are uniform squares with length $h=L/N_c$. The subsystem in each compartment is assumed to be homogeneous. The same types of particles and cells in different compartments are treated as different species; for example, we denote by $V_{i,j}$ the virus level in the compartment at location $(i,j)$ and consider $\{V_{1,1},\cdots, V_{1, N_c}, V_{2,1},\cdots, V_{2, N_c},\cdots, V_{N_c, N_c}\}$. Diffusion is treated as a reaction in which a molecule jumps to one of its neighboring compartments at a constant rate. Then with no-flux boundary conditions (or other conditions which depend on the experimental setting), diffusive jumps obey the following chain reactions for each $j\in\{1,2,\cdots, N_c\}$:

{
\begin{align*}
V_{1,j}\xrightleftharpoons[\rho_1]{\rho_1} V_{2,j}\xrightleftharpoons[\rho_1]{\rho_1}V_{3,j}\cdots \xrightleftharpoons[\gamma]{\rho_1}V_{N_c,j},\ V_{j,1}\xrightleftharpoons[\rho_1]{\rho_1} V_{j,2}\xrightleftharpoons[\rho_1]{\rho_1}V_{j,3}\cdots \xrightleftharpoons[\gamma]{\rho_1}V_{j,N_c},
\end{align*}}
where $\rho_1=d_V/h^2$. We assume that $D_{i,j}$ has similar chain reactions with $\rho_2=d_D/h^2$. We define the {\it propensity
function} for the jumps, for example, at the location $(i,j)$, for the four types of jumps (L: left, R: right, U: up, D: down) of virus:
$\alpha_{LV_{i,j}}(t)=\rho_1 V_{i,j}(t)$, $\alpha_{RV_{i,j}}(t)=\rho_1 V_{i,j}(t)$, $\alpha_{UV_{i,j}}(t)=\rho_1 V_{i,j}(t)$, and $\alpha_{DV_{i,j}}(t)=\rho_1 V_{i,j}(t)$. At the boundary, some jumping directions will not be considered for no-flux boundary conditions. 
For reactions, we assume that only molecules in the same compartment can react with each other. 

Different scenarios of virus production will contain different sets of reactions. 
In the first scenario, the reactions in the $(i,j)$ compartment are as follows:
\begin{flalign*}
&\phi \xrightarrow{\alpha_C C (1-C_T/K)} C,\ C \xrightarrow{\gamma_1V} C_V,\ C \xrightarrow{\gamma_2D} C_D,\ C_V \xrightarrow{\gamma_2D} C_{VD}, \ C_D \xrightarrow{\gamma_1V} C_{VD},\\
&C_V \xrightarrow{\nu_1} C_V^*,\ C_{VD} \xrightarrow{\nu_2} C_{VD}^*,\ D \xrightarrow{\delta_D} \phi,\ V \xrightarrow{\delta_V} \phi,\ \\
&C \xrightarrow{\delta_C} \phi,\ C_V \xrightarrow{\delta_{CV}} \phi,\ C_D \xrightarrow{\delta_{CD}} \phi,\ C_{VD} \xrightarrow{\delta_{CVD}} \phi,\ \\
&C_V^* \xrightarrow{\beta_1} (\alpha_1/\beta_1)V, \ C^*_{VD} \xrightarrow{\beta_{2}}  (\alpha_2/\beta_{2})V+ (\alpha_3/\beta_{2})D.
\end{flalign*}

In the second  scenario, the reactions (the first three rows of the previous scenario) are the same as the first one except for the production of viruses and DIPs. That is, we replace the last row by the following:
\begin{flalign*}
&\phi \xrightarrow {\alpha_1 C_V^*+\alpha_2 C_{VD}^*} V,  \phi \xrightarrow {\alpha_3 C_{VD}^*}D,\ C_V^* \xrightarrow{\beta_1} \phi,\ C^*_{VD} \xrightarrow{\beta_{2}} \phi.
\end{flalign*}

\section*{A new hybrid method for stochastic simulation}
In general, the computational cost for a stochastic simulation of a system in two-dimensional domain is extremely high. To reduce the computational cost and maintain the accuracy, we built up a new hybrid method which combines the advantages of our previous works: method of operator splitting~\cite{lo2016hybrid}, and spatially coupled hybrid method with adaptive interface~\cite{lo2019hybrid}.
In the new method, we use operator splitting to improve the efficiency and maintain the accuracy of the simulation; also, through this method with mixing stochastic and deterministic methods, we can apply the hybrid method for specific reactions while keeping others deterministic and hence consider only part of random effects to study which stochastic behavior plays an essential role in the pattern formation. 

The hybrid method combines two classes of simulation methods for modeling the reaction processes at two different scales. To capture the advantages of the methods with different scales, we use the method in our previous work~\cite{lo2019hybrid} to separate the spatial compartments into two types of regions with adaptive interfaces: 
1) the regions with ``large" numbers of molecules; 2) the regions with ``small" numbers of molecules. A more precise criteria for determining ``large" and ``small" will be given in Eq~(\ref{criterion}).

To better adapt to the complex system, we separate the compartments for each operator independently. That is, only the number of molecules of the species involved in an operator is considered in the regional division of that operator. 
We then apply SSA to approximate the dynamics in the region (1), and apply the PDE approximation in the region (2). For coupling two regions, we will apply the pseudo-compartment method~\cite{yates2015pseudo} with the adaptive interface method we used in~\cite{lo2019hybrid} in which the locations of the interfaces between different approaches are changing according to the distribution of molecules. With the idea of operator splitting mentioned above~\cite{lo2016hybrid}, our method can provide a numerical framework for studying the spatial stochastic effect of virus transmission. Through this new tool, we will have an efficient method to gain a quantitative understanding about the spatial effect of DIPs in virus transmission.

\subsection*{The domain and multiple interfaces for different reactions}
Consider a general reaction-diffusion system of $S$ species and $M$ chemical reactions and diffusion in 4 directions in a two-dimensional domain $\Omega$, which is partitioned into $N_c$ regular compartments of width $h$. Let $N_s(k,t)$ represent the amount of the $s$-th species in the $k$-th compartment at time $t$. Each compartment is small enough so individuals in it can be assumed well mixed.

The subdomain in which we employ the compartment-based regime for the $j$-th reaction at time $t$ is denoted by $\Omega_C^j(t)\subset \Omega$, and the other part of $\Omega$ that employs PDE is represented by $\Omega_P^j(t)$. $\Omega_C^j(t)$ contains all compartments in which the amount of at least one of the reactants in the $j$-th reaction is below the threshold value $\theta$.
To be specific, assume that reactants of the $j$-th reaction are $\{S_1, S_2,\cdots, S_m\}$. If 
\begin{eqnarray}\label{criterion}
\min_{i=1,2,\cdots,m}\{N_i(k,t)\}<\theta,
\end{eqnarray}
then the $k$-th compartment is assigned to the stochastic domain $\Omega_C^j(t)$, otherwise to the PDE domain $\Omega_P^j(t)$. 
In our algorithm, interfaces are adaptive. Domain division and multiple interfaces $I^j=\overline{\Omega_P^j(t)}\cap \overline{\Omega_C^j(t)}$ are updated every $\Delta t_I$. Fig~\ref{diagram} shows a one-dimensional illustration of the approach stated above.

\begin{figure}[!ht]
\includegraphics[width=0.85\textwidth]{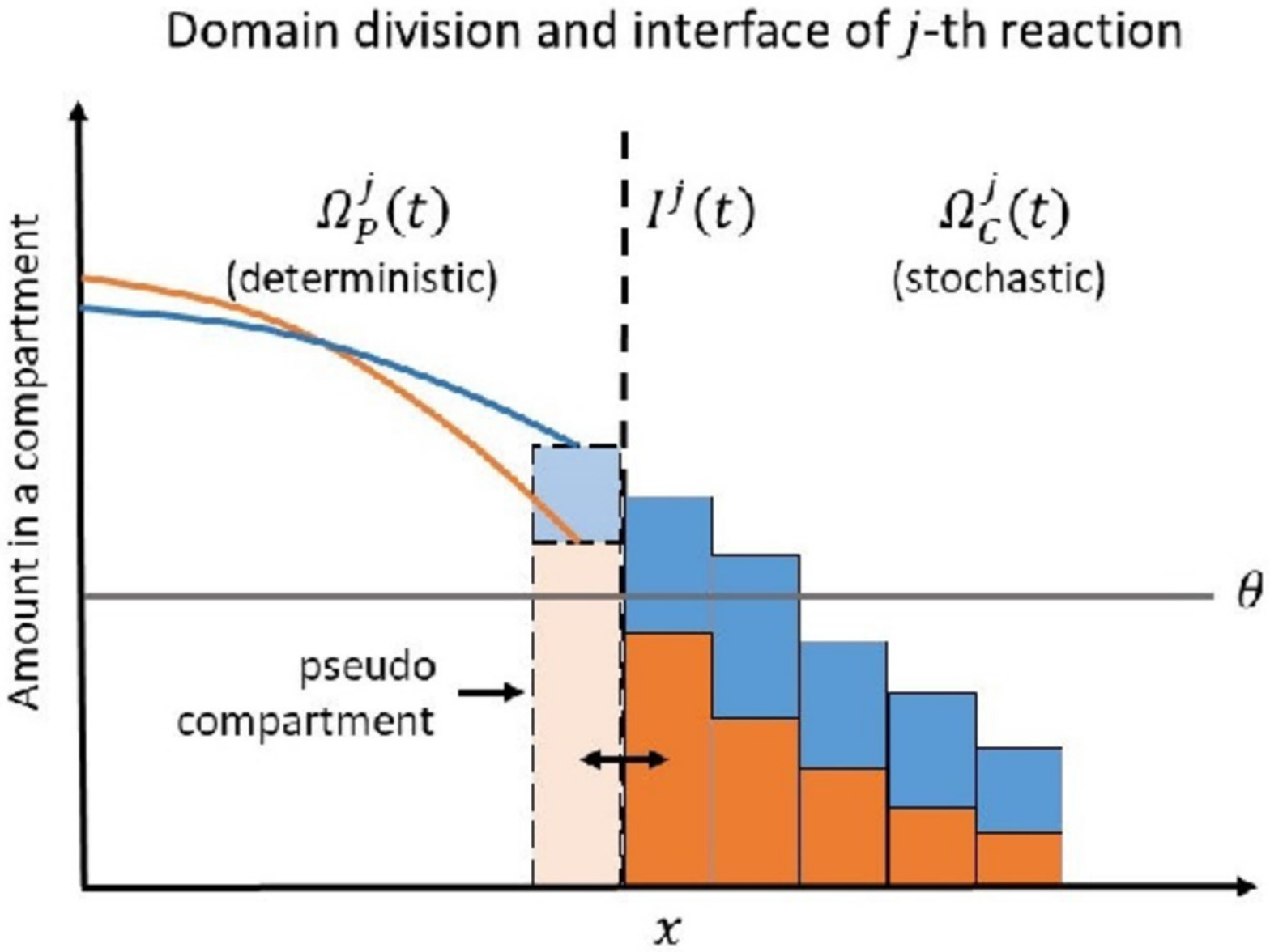}
\caption{{\bf An illustration of the domain division and interface of the $j$-th reaction.}}
Here we show an example with two reactants. The domain $\Omega_P^j(t)$ is modeled by PDE and the domain $\Omega_C^j(t)$ is modeled by compartment-based SSA. The amount of each species in a compartment in $\Omega_P^j(t)$ is $\int_{C_{k}}p_s(x,t)d x$.
If any reactant amount is below the threshold $\theta$, then that compartment is part of $\Omega_C^j(t)$.
Individuals can move between the boundary compartment of $\Omega_C^j(t)$ and the pseudo-compartment in $\Omega_P^j(t)$. In the two-dimensional case, diffusion takes four directions: up, down, left and right.
\label{diagram}
\end{figure}

It's worth noting that $I^{j_1}$ and $I^{j_2}$ can be the same if there is an inclusion relationship between the sets of reactants in the $j_1$-th and $j_2$-th reaction. In fact, the number of non-coincident interfaces is no more than the total number of species in the system. Therefore, compared with a single interface, multiple interfaces can capture stochastic fluctuations more accurately without increasing too much computation costs.

\subsection*{The pseudo-compartment method}
In this section, we will outline the pseudo-compartment method~\cite{yates2015pseudo}, which is the basis of our algorithm. In~\cite{yates2015pseudo}, Yates et al. introduced the pseudo-compartment method for diffusion. On this basis, we propose the possibility of multiple adaptive interfaces.

Consider a reaction-diffusion system of $S$ species, $M$ chemical reactions and diffusion in the four directions of the cross in a 2D domain $\Omega$.
In our algorithm, the PDE region varies for each reaction. So instead of just dividing the PDE based domain, we discretize the whole domain, $\Omega$, into a regular grid with spacing $\Delta x$. We consider the density of each species. For the $j$-th reaction at time $t$, the PDE numerical solution is updated for all grid points lie in $\Omega_P^j(t)$. Diffusion terms are treated in a similar way, but employing the implicit Euler method.
A zero-flux boundary condition is implemented in $\Omega_P^j(t)$, including domain boundaries and interfaces. Flux at the interface is implemented in the compartment-based regime.

The compartment-based regime evolves from the Gillespie algorithm (SSA). Consider the propensity function of reactions and diffusion, $\alpha_{i,j}(t)$, for compartment $C_i\subset\Omega_C^j(t)$. $\alpha_{i,j}(t)d t$ represents the probability that the $j$-th reaction (for $j\in\{1,\cdots, M+4\}$, including diffusion) occurs in $C_i$ during the small time interval $[t, t + d t]$. 

The coupling is implemented with a pseudo-compartment, $C_{-1}$, presented for diffusion between the deterministic and stochastic domains. This is a compartment adjacent to the interface but within deterministic domain $\Omega_P^j(t)$, where $j\in\{M+1,\cdots, M+4\}$, representing diffusion (four directions of the cross).
In order to correctly model the flux over the interface, individuals in the boundary compartment in $\Omega_C^j$ can jump into the pseudo-compartment with the usual diffusive rate, and vice versa.
The amount of each species within the pseudo-compartment is calculated through direct integration of the PDE,
\begin{eqnarray}
N_s(-i,t)=\int_{C_{-1,i}}p_s(x,t)d x,
\end{eqnarray}
where $p_s(x,t)$ is the PDE solution of density of the $s$-th species.
Then the propensity function for jumping from the pseudo-compartment to the adjoining compartment in $\Omega_C^j(t)$ is given by
\begin{eqnarray}
\alpha^*_{i,j}(t)=\frac{N_s(-i,t) D_s}{h^2}=\frac{D_s}{h^2}\int_{C_{-1,i}}p_s(x,t)d x, \quad j=M+1,\cdots,M+4.
\end{eqnarray}

The Gillespie's direct method~\cite{gillespie1977exact} is used to simulate the time evolution of stochastic regime. The time interval for next reaction, $\tau$, is determined by:
\begin{eqnarray}
\left\{\begin{aligned}
    &\alpha_0=\sum_{j=1}^{M+4}\sum_{C_i\in\Omega_C^j(t)}\alpha_{i,j}(t)+\sum_{j=M+1}^{M+4}\sum_{C_{-1,i}\in\Omega_P^j(t)}\alpha^*_{i,j}(t),\\
    &\tau=\frac{1}{\alpha_0}\ln{\frac{1}{r_1}},
    \end{aligned}
\ \right.
\label{tau}
\end{eqnarray}
where $r_1$ is a random variable uniformly distributed in $(0, 1)$. Use the SSA with the second random number $r_2$ to find the corresponding reaction or jump.
The algorithm then checks whether the closer update time is for PDE or SSA. If $t+\tau<t_P$, then the update is for SSA and $t=t+\tau$; otherwise it is for PDE and $t=t_P$, $t_P=t+\Delta t_P$.

\subsection*{Moving interface}
The multiple interfaces are updated with time step $\Delta t_I$, by recomparing amounts in a compartment of all reactants of each reaction with the threshold $\theta$. 
Similar to~\cite{spill2015hybrid}, after the interfaces are updated, we need to keep numbers in the stochastic domain are integer values, but we cannot simply get rid of the fractional parts.
Suppose the compartment $C_k$ is moved from the PDE domain to the stochastic domain, and the fractional part is
\begin{eqnarray}
P=\left\{\int_{C_{k}}p(x,t)dx \right\},
\end{eqnarray}
where $\{\cdot\}$ is the fractional part function. $P$ is used as the probability that an additional individual is kept in this compartment. We then take a uniform random number $r\in[0,1]$. If $r<P$ then we place the individual in compartment $C_k$; otherwise it is placed in the deterministic domain. 

The pseudocode for our algorithm is given in Algorithm~\ref{multi}.
\begin{breakablealgorithm}
\caption{} 
\label{multi}
\begin{enumerate}
    \item Initialize the time, $t = t_0$ and set the final time, $T$. Specify the PDE-update time step $\Delta t_P$ and initialize the next PDE time step to be $t_{P} = t + \Delta t_P$. Specify the interface-update time step $\Delta t_{I}$ and initialize the next interface-update time step to be $t_{I} = t + \Delta t_{I}$.
    \item Specify the PDE spacial step $\Delta x$ and the compartment width $h$. Initialize the amount of each species in each compartment, $N_s(k, t)$ for $k\in \{1,\dots,K\}$ and specify the threshold $\theta$. Compute the density, $p_s(x, t)=N_s(k, t)/h$ for PDE grid points.
    \item Determine the initial interface for each reaction $j$, $j\in\{1,2,\cdots, M\}$:
        \label{interface}
        \begin{enumerate}
            \item Find all $k$ such that $\min_{s\in S_j}\{N_s(k,t)\}<\theta$, where $S_j$ contains all species involved in reaction $j$, then the $k$-th compartment is part of the stochastic domain $\Omega^j_{C}$, and otherwise part of the PDE domain $\Omega^j_{P}$.
            \item All compartments adjacent to $\Omega^j_{C}$ (no diagonal angles) are regarded as pseudo compartments.
        \end{enumerate}
    \item Determine the time for the next ‘compartment-based’ event according to the Gillespie algorithm, $t_C = t + \tau$.
        \label{return}
    \item If $\min\{t_C, t_P,t_{I}\}=t_C$ then the next compartment-based event occurs:
        \begin{enumerate}
            \item Determine which event occurs according to the Gillespie algorithm.
            \item If the event is moving from stochastic domain to a pseudo compartment, $C_{-1}$, then for the corresponding $(s,k)$, $N_s(k,t+\tau)=N_s(k,t)-1$ and $p(x,t+\tau)=p(x,t)+\mathbb I_{[x\in C_{-1}]}/h$. Here, $\mathbb I_{[x\in A]}$ is an indicator function that takes the value 1 when $x\in A$ and 0 otherwise.
            \item If the event is moving from a pseudo compartment $C_{-1}$ to stochastic domain and $p(x, t)>1/h$ for all $x\in C_{-1}$, then $N_s(k,t+\tau)=N_s(k,t)+1$ and $p(x,t+\tau)=p(x,t)-\mathbb I_{[x\in C_{-1}]}/h$.
            \item Update the density for the pseudo compartment.
            \item Update the current time, $t = t_C$.
        \end{enumerate}
    \item If $\min\{t_C, t_P,t_{I}\}=t_{P}$ then the PDE domain is updated:
        \begin{enumerate}
            \item Apply backward Euler for diffusion terms and forward Euler for reaction terms.
            \item Update the density for the pseudo compartment.
            \item Update the current time, $t = t_{P}$ and set $t_{P} = t + \Delta t_P$.
        \end{enumerate}
    \item If $\min\{t_C, t_P,t_{I}\}=t_{I}$ then the interfaces are updated, similar to step \ref{interface}:
        \begin{enumerate}
            \item For each reaction, find all $k$ such that $\min_{s\in S_j}\{N_s(k,t)\}<\theta$, where $S_j$ contains all species involved in reaction $j$, then the $k$-th compartment is part of the stochastic domain $\Omega^j_{C}$, and otherwise part of the PDE domain $\Omega^j_{P}$.
            \item All compartments adjacent to $\Omega^j_{C}$ (no diagonal angles) are regarded as pseudo compartments.
            \item For the compartment $C_k$ that change from PDE domain to stochastic domain, let $P_s=\{\int_{C_k}p_s(x,t)dx\}$. Take a random number $r_s\in[0,1]$. 
            \begin{itemize}
            	\item If $r_s<P_s$ then $N_s(k,t_I)=$ ceil $(N_s(k,t))$ and $p_s(x,t_I)=p_s(x,t)- (1-P_s)/$ Area $(\Omega_{P}^j)$ for $x\in \Omega_{P}^j$;
            	\item otherwise, $N_s(k,t_I)=$ floor $(N_s(k,t))$ and $p_s(x,t_I)=p_s(x,t)+ P_s/$ Area $(\Omega_{P}^j)$ for $x\in \Omega_{P}^j$.
            \end{itemize}
            \item Update the current time, $t = t_{I}$ and set $t_{I} = t + \Delta t_I$. 
        \end{enumerate}
    \item If $t \leq T$, return to step \ref{return}.\\
          Else end.
\end{enumerate}
\end{breakablealgorithm}

\subsection*{Parameter estimation}
We consider the spatial domain as a two-dimensional square with length $L=2.552$mm, which is the same as the experimental data; for the PDE numerical scheme, we apply the central difference scheme to discretize the Laplace operation with $\Delta x =\Delta y = 0.058$mm; for the temporal numerical scheme, we use the backward Euler method for the Laplace operation and forward Euler method for the other terms with time step $\Delta t=0.01$h. In the SSA approximation, the domain is partitioned into square compartments with dimension $h \times h=\Delta x \times \Delta y$.

Diffusion coefficients of virus and DIP are set to be $2.38 \times 10^{- 6}$cm$^2$/h in~\cite{Clark2011} while the decay rate is $4.0 \times 10^{-5}$s$^{-1}$. As the diffusion rate varies according to the environment and plays a vital role in spatial distribution, we increased the former $d_V=d_D=2.38 \times 10^{- 3}$mm$^2$/h to match the experimental data and left the latter unchanged $\delta_V=\delta_D=0.144$h$^{-1}$.

In~\cite{Mapder2019}, the rate of virus production is expressed as the product of the number of viruses released per cell after packaging and the rate at which each cell produces viruses. Therefore $\alpha_1=758.045\times(68.503\times10^{\pm2}$d$^{-1}$$) =2163.682\times10^{\pm2}$h$^{-1}$, and $\alpha_3=38.259\times (21.782 \times 10^{\pm2}$d$^{-1}$$)=34.723\times10^{\pm2}$h$^{-1}$. The wide range of parameters allows us to choose a suitable value to match the experimental results. So we set $\alpha_1=6.491$h$^{-1}$ and $\alpha_3=69.446$h$^{-1}$. Since DIPs may exhibit a replication advantage over infectious viruses~\cite{Baltes2017}, we assumed $\alpha_2=\alpha_3/10$ in this work.

Same as~\cite{Mapder2019}, the intrinsic rate of uninfected cell proliferation $\alpha_C=15.217$d$^{-1}$$=0.634$h$^{-1}$. But the cellular carrying capacity of proliferation varies depending on the experimental environment. We let $K=3.505\times 10^5\times h^2$ cells/compartment to match the experimental data, where $h^2$ is the compartment area.

The rate of maturation of $C_V$ cells into $C_V^*$ cells is $9.863\times10^{\pm2}$d$^{-1}$ in~\cite{Mapder2019}. We slightly increase it to $\nu_1=\nu_2=0.205$h$^{-1}$ because mature infected cells are observed later in experiments. 
$\beta_1$ and $\beta_2$ are considered as the death rate of $C_V^*$ and that of $C_V^*$ respectively, which are $2.426\times10^{\pm2}$d$^{-1}$ in~\cite{Mapder2019}. We take $\beta_1=\beta_2=0.05$h$^{-1}$ in simulations.

Virus and DIP infection rate is $2.45 \times 10^{- 10}$d$^{-1}$$=1.02 \times 10^{- 11}$h$^{-1}$ in~\cite{Mapder2019}, which is relatively small. Different experiments and higher cell density may lead to a larger infection rate. Hence we set $\gamma_1=\gamma_2=4\times 10^{-4}$h$^{-1}$.

The infected cell death rate is $5.91\times10^{-2}$h$^{-1}$ in~\cite{Clark2011}, which is used as death rates for all cells in our simulations.

All parameters are listed in Table~\ref{table1}. It is worth noting that our set of parameters can guarantee that species in the system without DIPs will coexist in the following simulations. A detailed proof is provided in Appendix.

\begin{table}[!ht]
\begin{adjustwidth}{-2.25in}{0in} 
\centering
\caption{{\bf Parameters used in the simulations.}}
\begin{tabular}{|l|l|l|}
    \hline
    {\bf Parameter} & {\bf Definition} & {\bf Value}\\ \thickhline
    $d_V$  & Diffusion coefficient of virus  & $2.38\times 10^{-3}$mm$^2$/h \\\hline
    $d_D$  & Diffusion coefficient of DIP  & $2.38\times 10^{-3}$mm$^2$/h \\\hline
    $\alpha_1$ & Rate of virus production from virus-infected cell & $6.491$h$^{-1}$ \\\hline
    $\alpha_2$ & Rate of virus production from co-infected cell & $\alpha_3/10$ \\\hline
    $\alpha_3$ & Rate of DIP production from co-infected cell  & $69.446$h$^{-1}$ \\\hline
    $\alpha_C$ & Rate of uninfected cell proliferation  & $0.634$h$^{-1}$  \\\hline
    $K$        & Cellular carrying capacity of proliferation & $3.505\times 10^5\times h^2$cell/compartment\\\hline
    $\nu_1$    & Rate of maturation of $C_V$ cells into $C_V^*$ cells & $0.205$h$^{-1}$\\\hline
    $\nu_2$    & Rate of maturation of $C_{VD}$ cells into $C_{VD}^*$ cells & $0.205$h$^{-1}$  \\\hline
    $\beta_1$  & Death rate of $C_V^*$  &0.05 \\\hline
    $\beta_2$  & Death rate of $C_{VD}^*$  &0.05 \\\hline
    $\gamma_1$ & Virus infection rate & $4\times 10^{-4}$h$^{-1}$ \\\hline
    $\gamma_2$ & DIP infection rate & $4\times 10^{-4}$h$^{-1}$ \\\hline
    $\delta_V$ & Virus decay rate & $0.144$h$^{-1}$ \\\hline
    $\delta_D$ & DIP decay rate & $0.144$h$^{-1}$ \\\hline
    $\delta_{i},i=C,C_V,C_D,C_{VD}$ & Death rate of cells & $0.059$h$^{-1}$\\\hline
\end{tabular}
\label{table1}
\end{adjustwidth}
\end{table}

\subsection*{Interpretation of experimental data}
The experiment data published in~\cite{Baltes2017} is composed of time series of images obtained via microscopy from the co-propagation of infectious and defective viruses in a population of biological cells. These co-infection experiments were initiated with the same  virus inputs (MOI 30) but different DIP inputs (namely MOI 0,1,10 and 84). and microscopy images were taken at 7 hours, 13 hours, 19 hours and 25 hours post infection. The DIP expresses a green fluorescent protein (GFP) and the wild-type virus expresses a red fluorescent protein (RFP)  
There are three to five time series for each of the RFP intensity and the GFP intensity. Each image has size of $(2200,2200)$ with diameter of $1.16\mu$m pixel. The scale bar is $0.5$mm.

Fluorescent protein labeling is usually used for qualitative purposes, and there is no linear relationship between brightness and intensity. Therefore the experimental images only provide a reference for virus expression in simulations.

Since in the following simulations we employed the compartment-based method while experiments provide scatter diagrams, we have done some preprocessing to compare them with the computer simulation results. Fig~\ref{lab}A is a representative experiment figure. We extracted the red single channel (the virus is expressed) and filtered noise, as shown in Fig~\ref{lab}B. We then did morphological transformations (dilation followed by erosion) to close small holes inside the objects. Therefore Fig~\ref{lab}C maintains the critical features of virus expression in experiments and is more approximate to compartment-based.
\begin{figure}[!ht]
\includegraphics[width=0.85\textwidth]{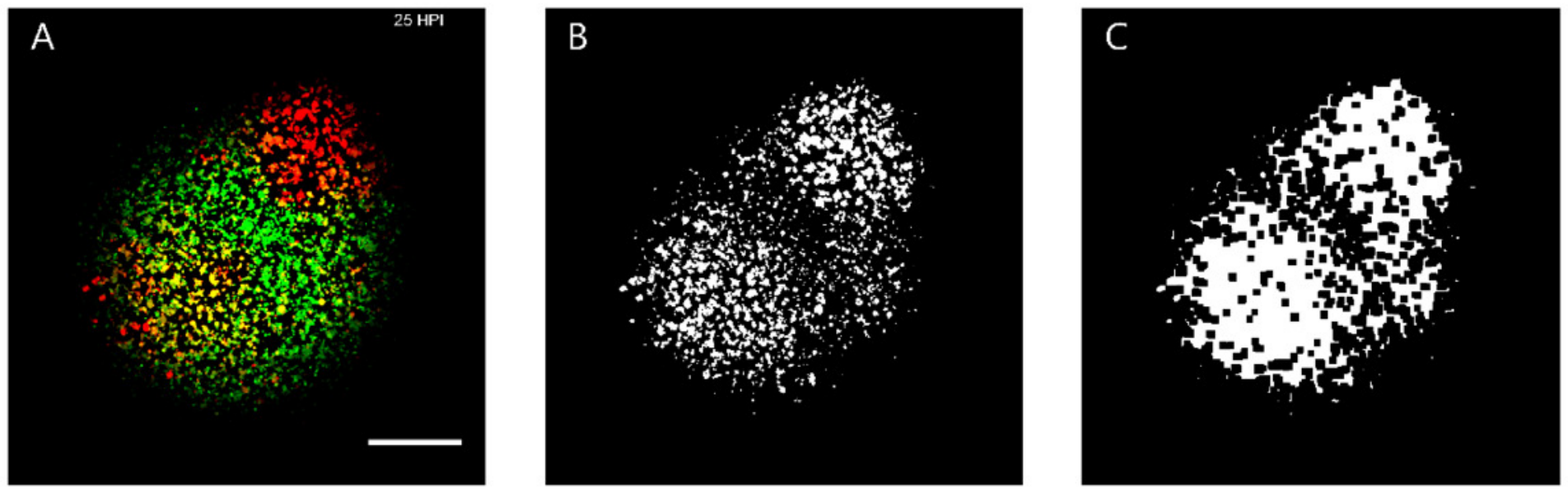}
\caption{{\bf A representative example of interpreting the preprocessing of experimental figures.}}
\label{lab}
\end{figure}


\section*{Results}
\subsection*{Dynamics and pattern formation of virus expression}
We first study the PDE model in Eqs~(\ref{equation1})-Eqs~(\ref{equation1_1}). Fig~\ref{PDE_t}A shows the time series of $C^*_V$ and $C^*_{VD}$ spatial distribution in a 2D domain at time 9-25h with no DIP, as well as images of experimental data at the same time in~\cite{Baltes2017}. In both simulations and experiments, viruses are uniformly radially distributed. In Fig~\ref{PDE_t}B, the initial conditions include $C_{VD}(0)=100$ for the PDE model, then viruses are distributed in a ring while the DIPs are radially distributed in the center. Compared with the experimental results under similar conditions, patchiness is not observed in PDE simulations. 

\begin{figure}[!ht]
\includegraphics[width=0.85\textwidth]{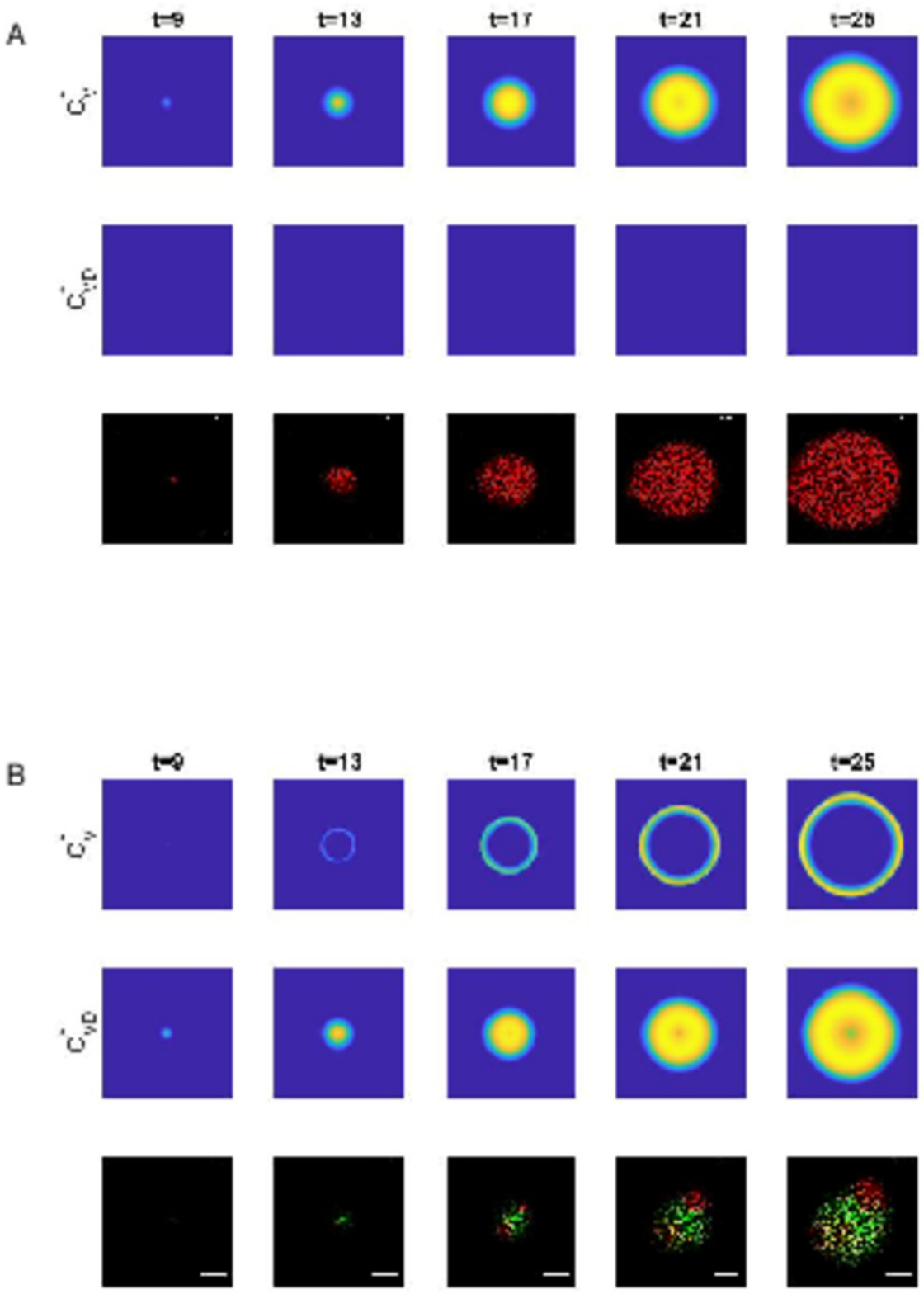}
    \caption{{\bf Dynamics of virus and DIP in cells in PDE simulations and experiments.}}
    A: Time series plot of virus in cells ($C^*_V$) and DIP in cells ($C^*_{VD}$) growth in PDE simulations and representative experiment with initial DIP equal to 0. B: Time series plot of virus in cells ($C^*_V$) and DIP in cells ($C^*_{VD}$) growth in PDE simulations and representative experiment with initial DIP equal to 84.
    \label{PDE_t}
\end{figure}

In stochastic simulations, the same types of particles and cells in different compartments are treated as different species, for example, for $V$, denoted by $\{V_{1,1},\cdots, V_{1, N_c}, V_{2,1},\cdots, V_{2, N_c},\cdots, V_{N_c, N_c}\}$. The initial condition is $V_{i,j}(0) = D_{i,j}(0) = C^*_{Vi,j}(0) = C^*_{VDi,j}(0) = C_{Di,j}(0)= 0$ and $C_{i,j}(0) = 1000$ for all $(i, j)$,
$C_{Vi,j}(0) = C_{VD i,j}(0)=0$ for all $(i, j)$ except the midpoint $C_{V22,22}(0) = 100$, and $C_{VD 22,22}(0)$ varies from $0$ to $400$.

Two scenarios are considered in simulations and compared with experimental results:
\begin{enumerate}
    \item[] Scenario 1:$\;$ infected cells produce virus and DIPs through cell bursting;
    \item[] Scenario 2:$\;$ infected cells keep producing viruses and DIPs continuously.
\end{enumerate}
Fig~\ref{case1} shows the evolution of the virus without initial DIP. The first two rows are time series of amounts of matured infected cells $C^*_V$ and $C^*_{VD}$, which is proportional to viral expression and DIP—virus expression in Scenario 1 from time $t=9$h to $25$h. The third and fourth rows are those in Scenario 2. The experimental observation has an inherent threshold, and images have been denoised; therefore, we also introduced a cut-off for simulation data. Namely the amount of cells is set to be zero if it is less than the cut-off value $50$, which is also applied to all the following simulations. The last row is the evolution of a representative experiment without initial DIP. When there is no DIP in the system initially, there is no DIP growth, and the virus growth is radially symmetric and flat in both scenarios and experiments. 

\begin{figure}[!ht]
\includegraphics[width=0.85\textwidth]{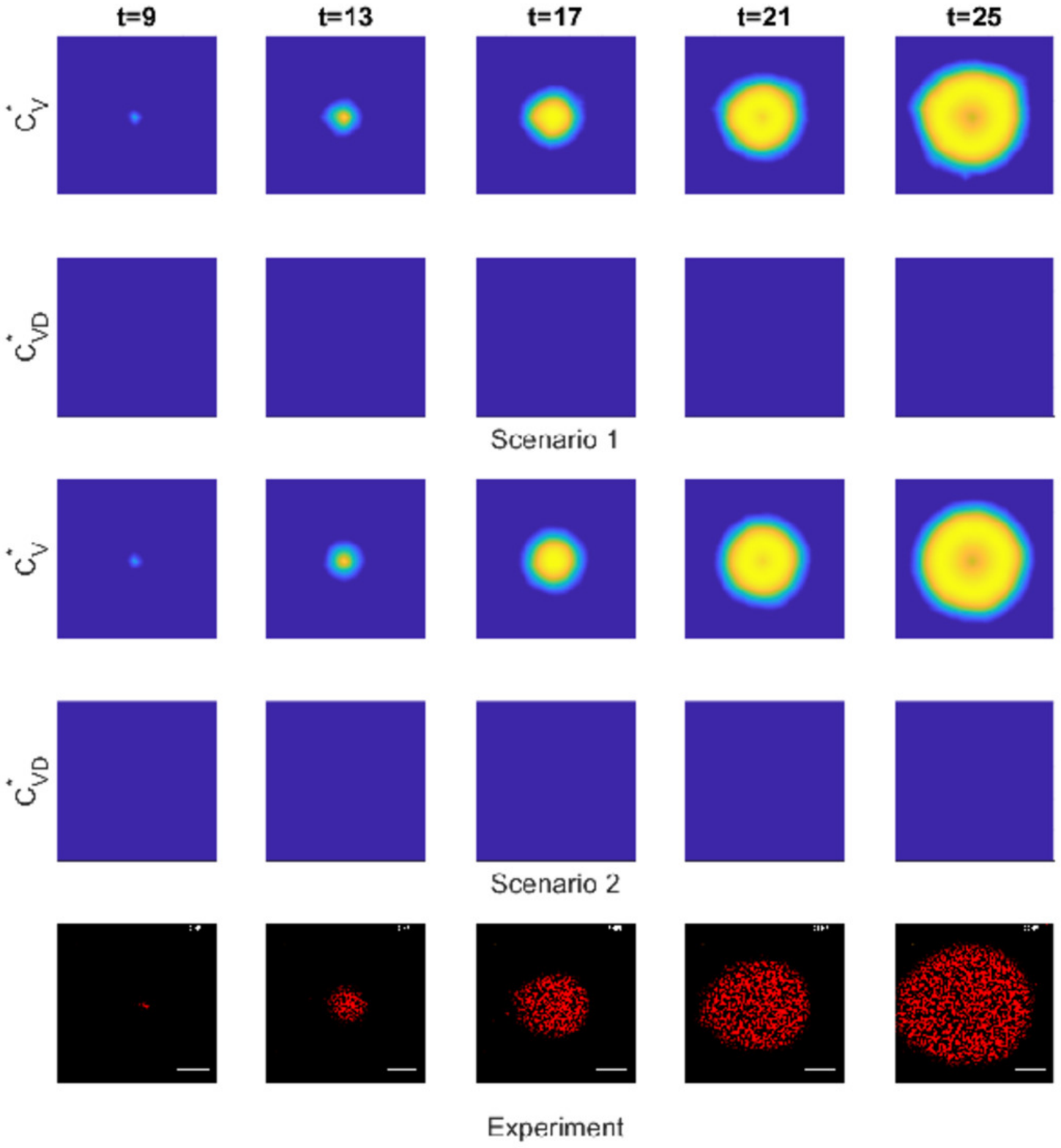}
\caption{{\bf Dynamics of virus and DIP in cells in 2 scenarios simulations and representative experiment with initial DIP equal to 0.}}
Row 1 and 2 are time series plots of virus in cells ($C^*_V$) and DIP in cells ($C^*_{VD}$) growth in Scenario 1 (infected cells produce virus and DIPs through cell bursting); 
Row 3 and 4 are time series plots of those in Scenario 2 (infected cells keep producing viruses and DIPs); Row 5 is the representative experimental results.
\label{case1}
\end{figure}

In experiments, the radial symmetry disappears as the initial amount of DIP increases. In fact, patchy formation is sensitive to the dose of DIP. It can be observed even with a small initial dose of DIP (Fig~\ref{case2}). When the initial DIP is raised from 10 to 84 in experiments, the majority of the virus at the end is concentrated (see Fig~\ref{case3}, Fig~\ref{case4}). Simulations show similar results, but Scenario 1 shows a much higher probability of forming a pattern than S2 and matches the experimental data better. In scenario S1, infected cells produce viruses and DIPs through cell bursting, and then viruses and DIPs diffuse, which leads to a more accidental position; while in Scenario 2, infected cells keep producing viruses and DIPs, meaning the location of those cells will continuously produce virus and DIPs. Hence the spatial distribution is more centralized rather than patchy.

\begin{figure}[!ht]
\includegraphics[width=0.85\textwidth]{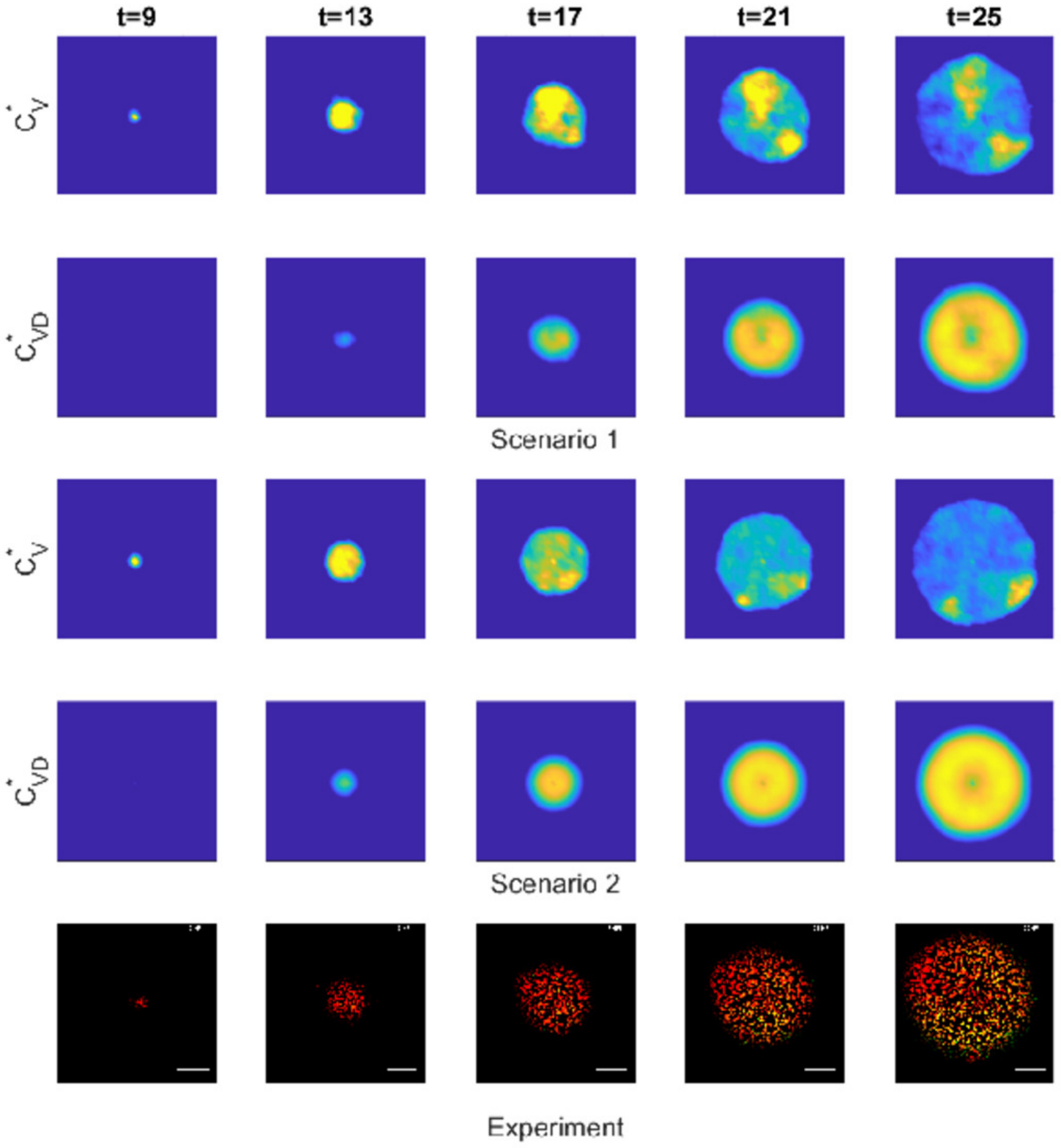}
\caption{{\bf Dynamics of virus and DIP in cells in 2 scenarios simulations with $C_{VD22,22}(0)=4$ and representative experiment with initial DIP equal to 1.}}
Row 1, 2 are time series plots of virus in cells ($C^*_V$) and DIP in cells ($C^*_{VD}$) growth in Scenario 1 (infected cells produce virus and DIPs through cell bursting); 
Row 3, 4 are time series plots of those in Scenario 2 (infected cells keep producing viruses and DIPs); Row 5 is the representative experimental results.
\label{case2}
\end{figure}

\begin{figure}[!ht]
\includegraphics[width=0.85\textwidth]{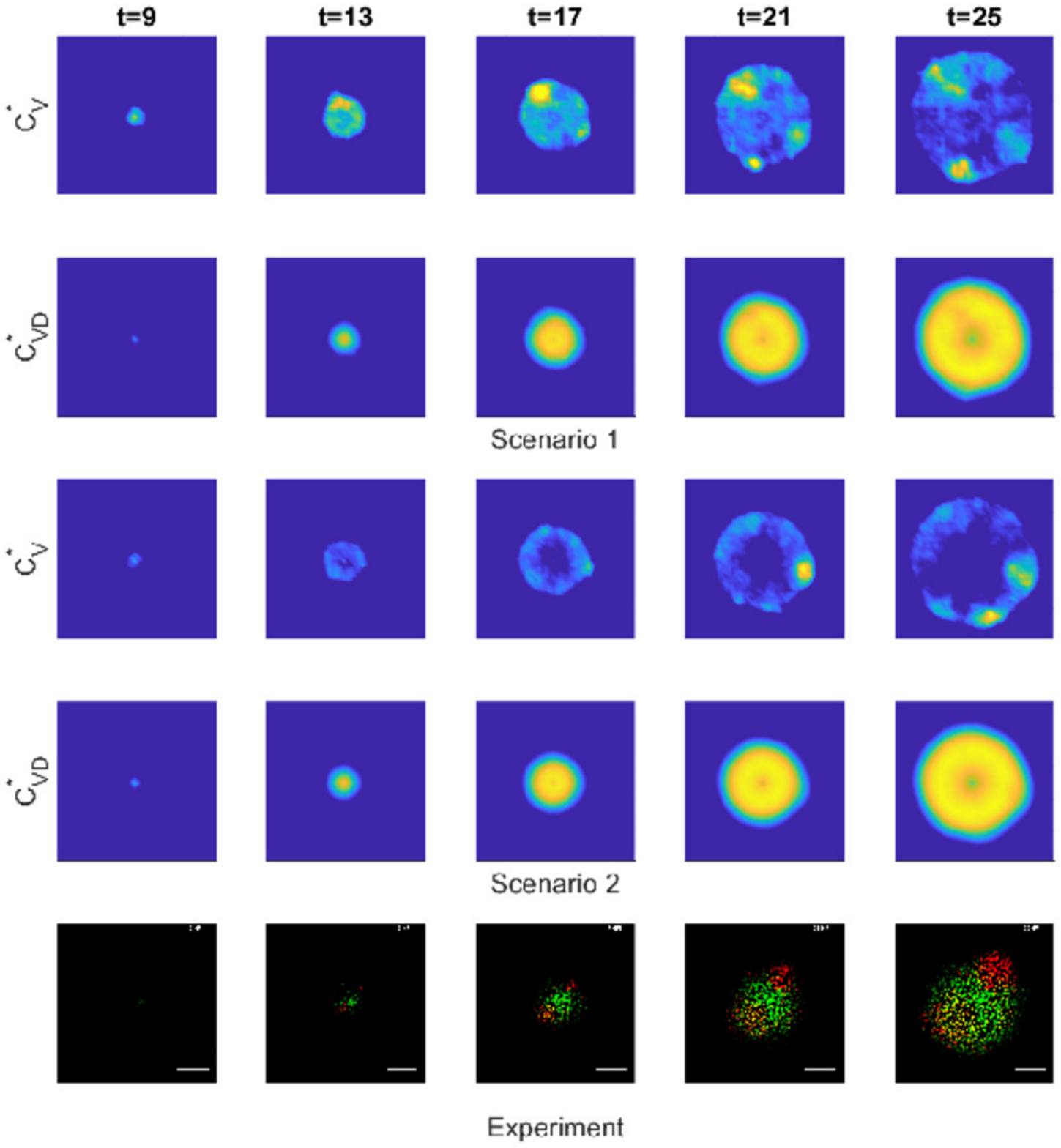}
\caption{{\bf Dynamics of virus and DIP in cells in 2 scenarios simulations with $C_{VD22,22}(0)=40$ and representative experiment with initial DIP equal to 10.}}
Row 1, 2 are time series plots of virus in cells ($C^*_V$) and DIP in cells ($C^*_{VD}$) growth in Scenario 1 (infected cells produce virus and DIPs through cell bursting); 
Row 3, 4 are time series plots of those in Scenario 2 (infected cells keep producing viruses and DIPs); Row 5 is the representative experimental results.
\label{case3}
\end{figure}

\subsection*{Spread rate of virus}
To quantify the spread characteristics of viral expression under stochastic effects, we define the virus radius as:
\begin{eqnarray}
R(t)=\sqrt{\frac{Area(C_V(x,y,t))}{\pi}},
\end{eqnarray}
where $C_V(x,y,t)$ represents the amount of cells infected by virus at grid point $(x,y)$ at time $t$.
Since a certain amount of virus expression is required to be observed in the experiment and noise is filtered, we also set a cut-off $\theta=50$ for computing area, i.e.
\begin{eqnarray}
Area(C_V(x,y,t))=\sum_{(x,y)} \mathbb I_{C_V(x,y,t)>\theta} \Delta x\Delta y.
\end{eqnarray}
For experimental data, a detailed illustration is in Fig~\ref{lab}. Fig~\ref{radius} shows the radius of virus versus time $9\leq t\leq 25$ (h) with different initial DIP inputs in 2 scenarios simulations and experiments. We can see the radius keeps increasing and viruses keep spreading in all cases. Whereas as initial DIPs increase, in both experiments and simulations, the growth rate of radius goes down, which is due to the inhibitory effect of DIPs on viruses. On the other hand, the Scenario 1 can better match the experimental results, both in terms of the dynamics and the level of fluctuations.
\begin{figure}[!ht]
\includegraphics[width=0.85\textwidth]{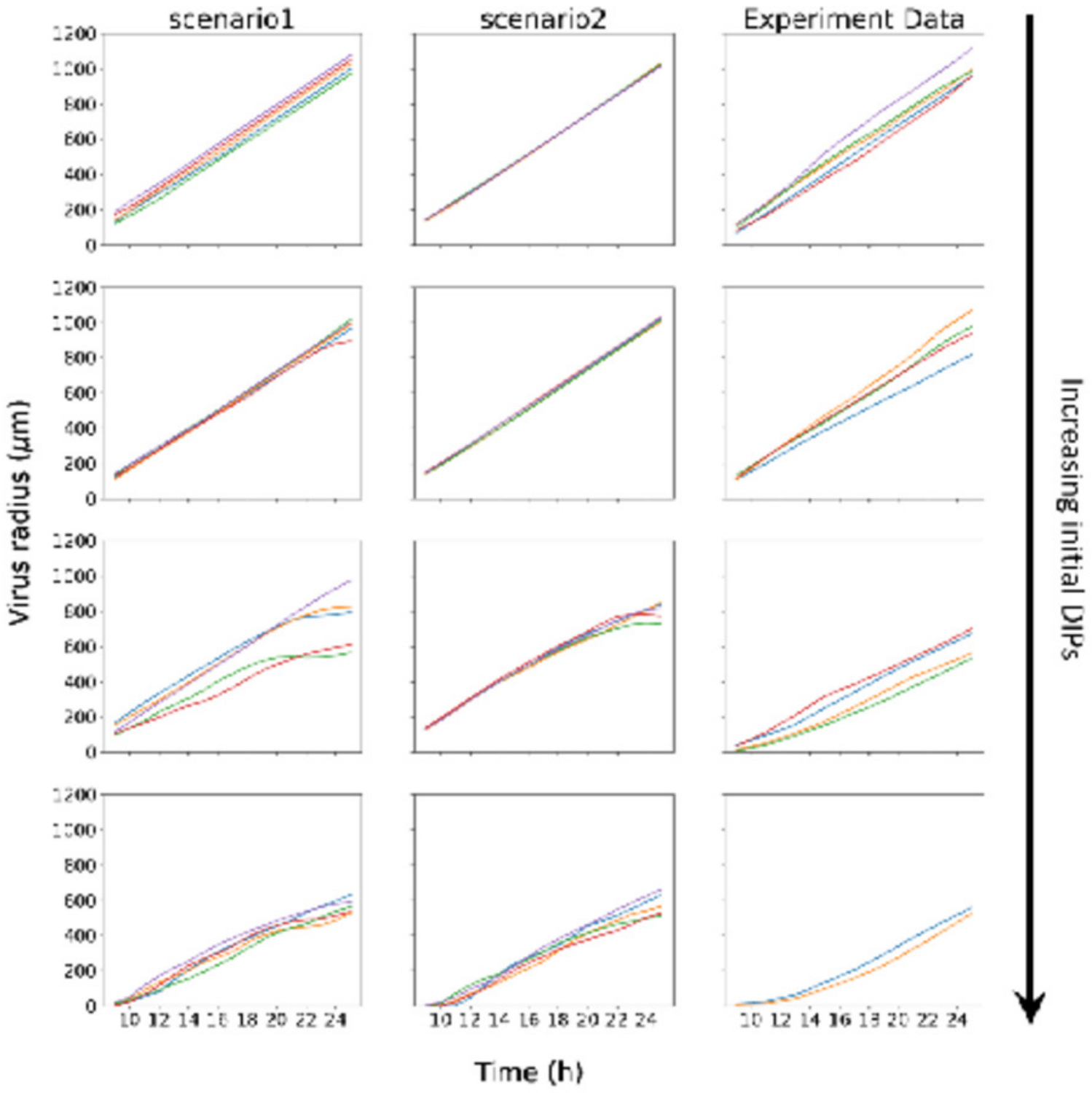}
\caption{{\bf Radius of virus against time with different initial DIP inputs in 2 scenarios simulations and experiments.}}
\label{radius}
\end{figure}

We note that after a certain time, the plague radius grows linearly with respect to time for each
fixed initial dose of DIPs, and studied the relationship between the virus radius growth rate and initial DIP dose. To get rid of the difference in units of initial conditions in simulations and experiments, we consider a dimensionless ratio $\rho=\frac{C_{VD}(0)}{C_{V}(0)}$. Since initial viruses remain the same, $\rho$ is proportional to initial DIPs.
Fig~\ref{vs_initial} shows the relationship between the virus radius growth rate and initial DIP dose intuitively. We used a logarithmic x-axis, so it is shifted by 0.01 to avoid troubles when $\rho=0$. The growth rate was computed using the data points after $t = 13$h to ensure in all cases we have close to linear growth in radius vs. time (slope of lines in Fig~\ref{radius}).
We run 50 group simulations for each initial condition for computing the average.
When there is no DIP in the system, the virus radius growth rates are the same in both scenarios; as initial DIPs increase, the growth rate drops dramatically, which means DIPs slow down the growth of virus particles.
\begin{figure}[!ht]
\includegraphics[width=0.85\textwidth]{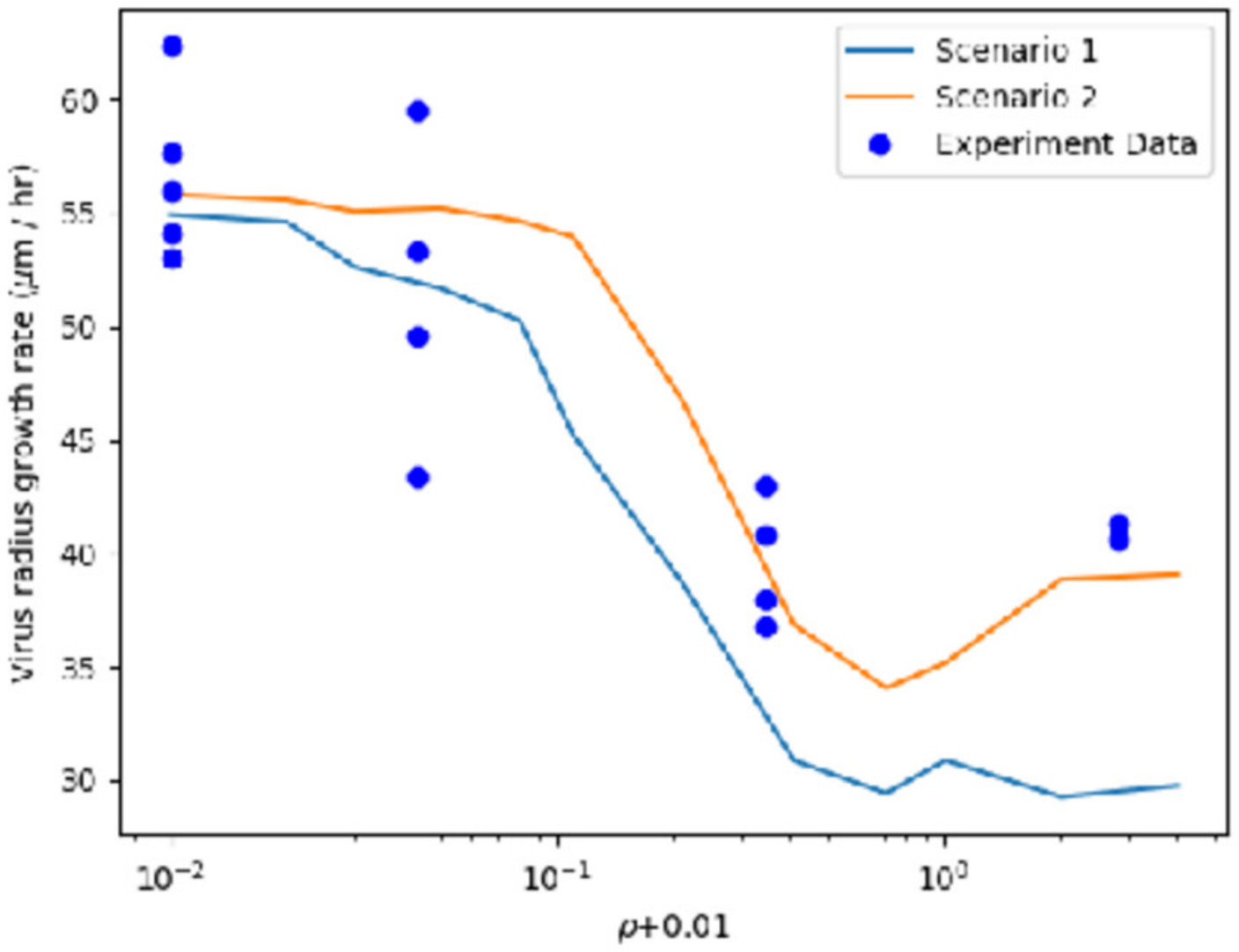}
\caption{{\bf The growth rate of virus radius against initial DIP inputs in 2 scenarios simulations and experiments.}}
\label{vs_initial}
\end{figure}

\subsection*{Patchiness via $q$-statistic} 
Patchy spatial patterns are typically observed in the image data when the initial dose of DIP is large enough. Therefore, we quantify the patchiness of image data by the $q$-statistic, which is a standard spatial statistic used to measure spatial stratified heterogeneity~\cite{wang2016measure}. The definition of this statistic depends on our choice of strata, which is a decomposition of the image data. In our case, the entire image is divided into 30 sectors with an equal ratio of the angle to form 30 strata $S = \{L_1, L_2,\cdots, L_{30}\}$ and the union of $L_i$ is the whole plaque $P$. A visual illustration is as shown in Fig~\ref{stratas}. In our case, since experiments employed qualitative rather than quantitative methods, that is, we can see viral expression at all fluorescent points but the brightness of these points is not proportional to the intensity. So we convert all figures binary and $M(i,j,t)$ denotes the brightness at the $(i,j)$-th pixel at time $t$ for the image, which range is $\{0, 255\}$. For simulation results, we set a threshold for the binary transformation to approximate the threshold inherent in the experimental methods and offset the loss when denoising the experimental images. Specifically, when $C_{V}^*<50$, $M=0$ and the point is black in the image; when $C_{V}^*\geq50$, $M=255$ and the point is red (an example in Fig~\ref{simu_binary}).
Then, the $q$-statistics is defined to be 
\begin{eqnarray}
\label{qstat}
\begin{aligned}
q_t &= 1 - \frac{\sum_{L \in S}\sum_{(i,j) \in L}\left(M(i,j,t) - \overline{M_L^t}\right)^2}{\sum_{(i,j) \in P}\left(M(i,j,t) - \overline{M_P^t}\right)^2}\\
&= 1 - \frac{1}{N \sigma^2} \sum_{L \in S} N_L \sigma_L^2,
\end{aligned}
\end{eqnarray}
where $\overline{M_A^t} = \frac{\sum_{(i,j) \in A} M(i,j,t)}{|A|}$ and $|A|$ is the cardinality of set $A$. $N$, $N_h$, $\sigma$, $\sigma_h$ denote the number of pixels in the entire image, the number of pixels in each stratum, the standard deviation of the entire image and the standard deviation of each stratum, respectively. This statistic is invariant under spatial scale and remains the same if the intensity of the image is multiplied by a factor. 

A more intuitive formula for the $q$-statistic is as follows~\cite{wang2016measure}, here we omit the time dependence, meaning denote $M(i,j,t)$ by $M_{i,j}$ and $q_t$ by $q$; The denominator of Eqs~(\ref{qstat}) can be written as 
\begin{eqnarray}
    \sum_{L \in S}\sum_{(i,j) \in L} \left( M_{i,j} - \overline{M_L}\right)^2 + \sum_{L \in S} |L|\left(\overline{M_L} - \overline{M_P}\right).
\end{eqnarray}
Call those two terms the sum of squares within strata (SSW) and the sum of squares between strata (SSB) and note that the numerator of Eqs~(\ref{qstat}) is exactly SSW, so 
\begin{eqnarray}
    q = 1 - \frac{SSW}{SSW+SSB}.
\end{eqnarray}

So if $q \approx 1$, that means the sum of squares within strata is relatively more minor, indicating in each stratum, the virus is concentrated and the sum of squares between strata is somewhat more significant, meaning the differences between strata are large, so the image would appear to be more patchy. If $q \approx 0$, the variance in each stratum is large, and the differences between strata are minor, so the picture would appear to be not so patchy.

We study the behavior of $q$-statistic of cells infected by the virus at time $t=25$h, that is $C_{V}^*(25)$ in simulations, when the initial dose of DIPs varies. To get rid of the difference in units, we consider a dimensionless ratio $\rho=\frac{C_{VD}(0)}{C_{V}(0)}$. Since we always keep the initial viruses constant, $\rho$ is proportional to the initial dose of DIPs.

In Fig~\ref{q_vs_ini}, the x-axis is a logarithmic scale so we shift it by 0.01 to the right to avoid trouble when $\rho=0$ (initial DIP is zero). On the left, the blue line denotes the $q$-statistic of Scenario 1 (infected cells produce viruses and DIPs through cell bursting) while the green line denotes that of Scenario 2 (infected cells keep producing viruses and DIPs). The $95\%$ confidence intervals are also presented respectively. On the right, we marked the $q$-statistic for each experimental image at time $t=25$h and plot the average for four groups of experiments. Both scenarios show the same trend as experiments. When there is no DIP in the beginning, the $q$-statistic is minor, meaning the spatial distribution is uniform. The $q$-statistic increases as the initial dose of DIPs increases. Taking into account the conclusions of the previous section, DIPs slow down the growth of virus particles and make them more patchy. But when the initial dose of DIP is large enough, we observe a drop of $q$-statistic. It may be caused by the domination of DIPs. The $q$-statistic is sensitive to the changes in DIPs. On the other hand, Scenario 2 shows a closer magnitude of $q$-statistic to experimental data while that of Scenario 1 is relatively higher. When infected cells produce viruses and DIPs through cell bursting, their positions are more stochastic, and hence there is a larger probability of patchy formation, which also explains the wider confidence interval of Scenario 1.
\begin{figure}[!ht]
 \includegraphics[width=0.85\textwidth]{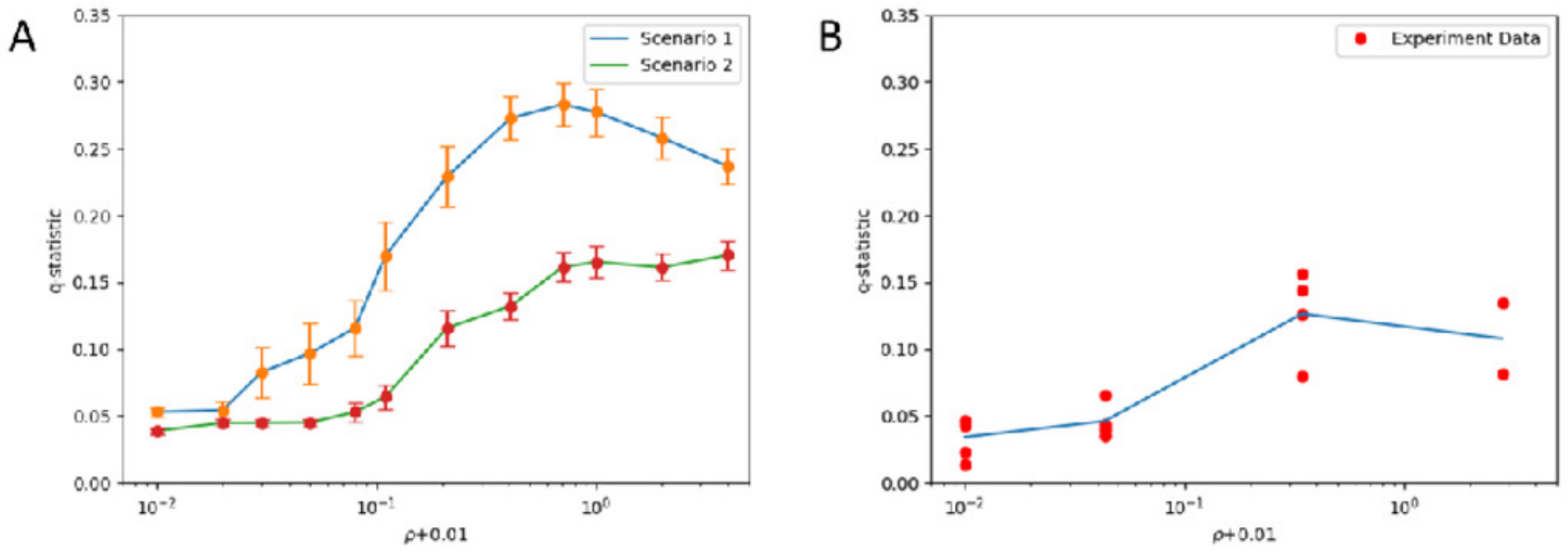}
	\caption{{\bf The average $q$-statistics of virus against initial DIP inputs at time t = 25h in 2 scenarios simulations and experiments.}}
	\label{q_vs_ini}
\end{figure}

\section*{Discussion}
DIPs can co-infect a cell with viable viruses and interfere with virus production~\cite{Baltes2017, Frensing2015}. However, the mechanism by which DIPs affect the spatial distribution of virus expression is still unclear.

In this work, we constructed a PDE model to describe the interaction between viruses and DIPs in a two-dimensional domain. Moreover, to study the stochastic effect on spatial dynamics of the virus spreading and patchy formation, we developed a stochastic reaction-diffusion system to describe the system in two different scenarios of virus production. In Scenario 1, infected cells produce viruses and DIPs through cell bursting. Therefore the position of virus production is accidental, which leads to a higher probability of patchy formation. In Scenario 2, infected cells keep producing viruses and DIPs. The virus is produced continuously at the cell position, making the spatial distribution concentrated in one point.
The patchy pattern observed in the experiments can be regenerated in our stochastic simulation results. Our model provides a good framework for studying reaction-diffusion systems under stochastic effects. 

We also built a hybrid algorithm for stochastic simulation. Classical stochastic methodologies are computationally intensive in two-dimensional cases. Our algorithm is based on the pseudo-compartment method~\cite{yates2015pseudo} and introduces adaptive multiple interfaces to adjust complex systems. 
It combines two scales of simulation methods for modeling the reaction processes and can capture the advantages of the methods with different scales. It improves computational efficiency and maintains critical stochastic features.
Our method can provide a numerical framework for studying the spatial stochastic effect of other biological systems and is compatible with different scale stochastic study methods like stochastic differential equations. 

We quantitatively studied the spread rate of the virus and showed the relationship between the spread radius growth rate and the initial dose of DIP. To measure the patchiness, we computed the $q$-statistic. Our simulations can simultaneously capture two spatial spread features (patchiness and spread rate) in wet-lab experimental data, which was not achieved in previous works. It supports that the DIPs slow down the growth of virus particles and make them more patchy. These quantitative methods and statistics are useful tools to understand and explain the diverse spatial-temporal features in complex biological systems.

\appendix
\renewcommand\thefigure{A.\arabic{figure}}    
\renewcommand\theequation{A.\arabic{equation}}
\section*{Appendix}
\setcounter{figure}{0}    
\setcounter{equation}{0}
\subsection*{Virus infection without DIPs}\label{A1}

\setcounter{equation}{0}
Setting $D=C_D=C_{VD}=C_{VD}^*=0$ in Eqs~(1) and Eqs~(2) gives the system
\begin{equation}
\begin{aligned}
\frac{\partial V}{\partial t}&=d_V\nabla^2V+\alpha_1C_V^*-\delta_V V,\\
\frac{\partial C}{\partial t}&=\alpha_C C(1-C_T/K)-\gamma_1CV-\delta_C C,\\
\frac{\partial C_V}{\partial t}&=\gamma_1CV-\nu_1C_V-\delta_{CV}C_V,\\
\frac{\partial C_V^*}{\partial t}&=\nu_1 C_V-\beta_1 C_V^*.
\end{aligned}
\label{eqn_withoutDIP}
\end{equation}

To study the dynamics of Eqs~(\ref{eqn_withoutDIP}), we obtain the following result for the homogeneous steady states of the system. Roughly, the number of steady states decreases from 3 to 1 as $\delta_C/\alpha_C$ increases pass two critical values.
\begin{lem}
 Let $Z = \frac{\delta_V\beta_1(\delta_{CV}+\nu_1)}{\gamma_1\alpha_1\nu_1}$. For studying the non-negative homogeneous steady states of Eqs~(\ref{eqn_withoutDIP}), there are three cases:
\begin{itemize}
\item[(1)] If $\alpha_C \le \delta_C$, there is only one steady state $E_0=(0,0,0,0)$.
\item[(2)] If $ \alpha_C\left(1-\frac{Z}{K}\right)\le \delta_C <\alpha_C$, there are two steady states, $E_0$ and $E_1=(0,\frac{K(\alpha_C - \delta_C)}{\alpha_C},0,0)$. 
\item[(3)]  If $\delta_C<\alpha_C\left(1-\frac{Z}{K}\right)<\alpha_C$, there are three steady states, $E_0$, $E_1$ and $E_2=(\overline{V}, Z, \overline{C}_V, \overline{C}_V^*)$ where
\begin{align*}
\overline{C}_V &= \frac{Z\left(\alpha_C\left(1-\frac{Z}{K}\right)-\delta_C\right)}{\nu_1 +\delta_{CV}+(\beta_1+\nu_1)\alpha_C Z/(\beta_1 K)},\\
\overline{C}_V^* &= \frac{\nu_1}{\beta_1}\overline{C}_V,\\
\overline{V}&= \frac{\alpha_1\nu_1}{\delta_V\beta_1}\overline{C}_V.
\end{align*}
\end{itemize}
\label{lemma 2.1}
\end{lem}

\begin{proof}
First, we consider the homogeneous steady state equations of Eqs~(\ref{eqn_withoutDIP}),
\begin{equation}
\begin{aligned}
0&=\alpha_1x_4-\delta_Vx_1,\\
0&=\alpha_C x_2(1-(x_2+x_3+x_4)/K)-\gamma_1x_2x_1-\delta_c x_2,\\
0&=\gamma_1x_2 x_1-\nu_1x_3-\delta_{CV}x_3,\\
0&=\nu_1 x_3-\beta_1 x_4.
\end{aligned}
\label{eqn_withoutDIP_ODE_ss}
\end{equation}
By the first and the last equations above, we obtain $x_1 =\alpha_1\nu_1x_3/(\delta_V\beta_1)$ and $x_4 = \nu_1 x_3/\beta_1$. Substitute $x_1 =\alpha_1\nu_1x_3/(\delta_V \beta_1)$  to the third 
equation, we have
$$0 =\gamma_1\alpha_1\nu_1x_2x_3/(\delta_V\beta_1)-\nu_1x_3-\delta_{CV}x_3,$$
$$0 =\gamma_1\alpha_1\nu_1(x_2-Z)x_3/(\delta_V\beta_1).$$
$$x_3=0 \text{ or }x_2=Z.$$

If $x_3 = 0$, we obtain that $x_1=x_4=0$. Consider the second equation with $x_1=x_3=x_4=0$, we have
$$ 0=(\alpha_C -\delta_C - \alpha_C x_2/K) x_2,$$
which leads to two possible non-negative solution $x_2 = 0$ or $x_2= K(\alpha_C-\delta_C)/\alpha_C$ if $\alpha_C >\delta_C$.

If $x_2=Z$, we consider the addition of the second and the third equations with $x_1 =\alpha_1\nu_1x_3/(\delta_V\beta_1)$ and $x_4 = \nu_1 x_3/\beta_1$. If $\delta_C<\alpha_C\left(1-\frac{Z}{K}\right)$, we can find a positive solution for $x_3$,
$$x_3 = \frac{Z\left(\alpha_C\left(1-\frac{Z}{K}\right)-\delta_C\right)}{\nu_1 +\delta_{CV}+(\beta+\nu_1)\alpha_C Z/(\beta_1 K)}.$$
We proved the existence of the homogeneous steady states for the three cases. 
\end{proof}
\begin{prop}
If $\alpha_C< \delta_C$, the solution will approach to $E_0=(0,0,0,0)$ as $t\rightarrow\infty$.
\end{prop}
\begin{proof}
By considering the second equations in Eqs~(\ref{eqn_withoutDIP}), 
$$\frac{\partial C}{\partial t}=\alpha_C C(1-C_T/K)-\gamma_1CV-\delta_CC\leq \alpha_C C-\delta_CC=(\alpha_C-\delta_C)C,$$
which leads to 
$$C(t,\vec{x}) \leq C(0,\vec{x})e^{(\alpha_C-\delta_C)t}.$$
If $\alpha_C-\delta_C<0$, then $C(t,\vec{x})$ will approach to zero as $t\rightarrow\infty$. 

Through considering the equations for $C_V$, $C_V^*$ and $V$ one by one, it is easy to show that when $C(t,\vec{x})$ approaches zero, $C_V$, $C_V^*$ and $V$ all approach zero as $t\rightarrow\infty$.  
\end{proof}
\begin{prop}
If $ \alpha_C\left(1-\frac{Z}{K}\right)< \delta_C <\alpha_C$ and $C(0,\vec{x})>0$, the solution will approach to $E_1=(0,\frac{K(\alpha_C - \delta_C)}{\alpha_C},0,0)$ as $t\rightarrow\infty$.
\end{prop}

\begin{proof}
Let $f(y_1,y_2) = (\alpha_C (1-y_1/K)-\gamma_1y_2-\delta_C)$. By $\alpha_C\left(1-\frac{Z}{K}\right)< \delta_C$, we observe that $f(y_1,y_2) <0$ for any $y_1>Z$ and $y_2>0$.  If $C\geq Z$, 
$$\frac{\partial C}{\partial t}=f(C_T,V)C\leq f(C,0)C<0.$$
It implies that $C(t,\vec{x})< Z$ as $t\rightarrow \infty$. 

Define $Y = (\nu_1+\delta_{CV})(\alpha_1C_V^*+ \beta_1 V)+\alpha_1\nu_1 C_V$. Consider the derivative of $Y$ with respect to $t$,
$$\frac{\partial Y}{\partial t} =\beta_1(\nu_1+\delta_{CV}) d_V\nabla^2V + (C-Z)V\alpha_1\nu_1\gamma_1.$$
Take integration over the spatial domain $\Omega$, we can obtain
$$\frac{\partial \int_\Omega Y dA}{\partial t} =\int_\Omega (C-Z)V\alpha_1\nu_1\gamma_1 dA.$$
If $V$ is not zero, the right-hand side becomes negative as $t\rightarrow \infty$, then $Y$ will decrease. 
If there exists $T>0$ such that $V$ becomes zero for $t>T$, it is easy to show that $C_V$ and $C_V^*$ approach to zero; if not, $Y$ will approach to zero and It implies that $V$, $C_V$ and $C_V^*$ approach to zero as $t\rightarrow \infty$. 
When $t\rightarrow \infty$, $V$, $C_V$ and $C_V^*$ approach to zero and
$$\frac{\partial C}{\partial t}=\alpha_C C(1-C/K)-\delta_C C$$
If $C$ is not zero with the condition $\alpha_C\left(1-\frac{Z}{K}\right)< \delta_C <\alpha_C$, the equation above implies that $C\rightarrow\frac{K(\alpha_C - \delta_C)}{\alpha_C}$ as $t\rightarrow \infty$.

\end{proof}
Our set of parameters does not meet the conditions of the propositions above. Therefore it can guarantee that species in the system will not become extinct in the following simulations.

\subsection*{Figures}
\begin{figure}[H]
\includegraphics[width=0.85\textwidth]{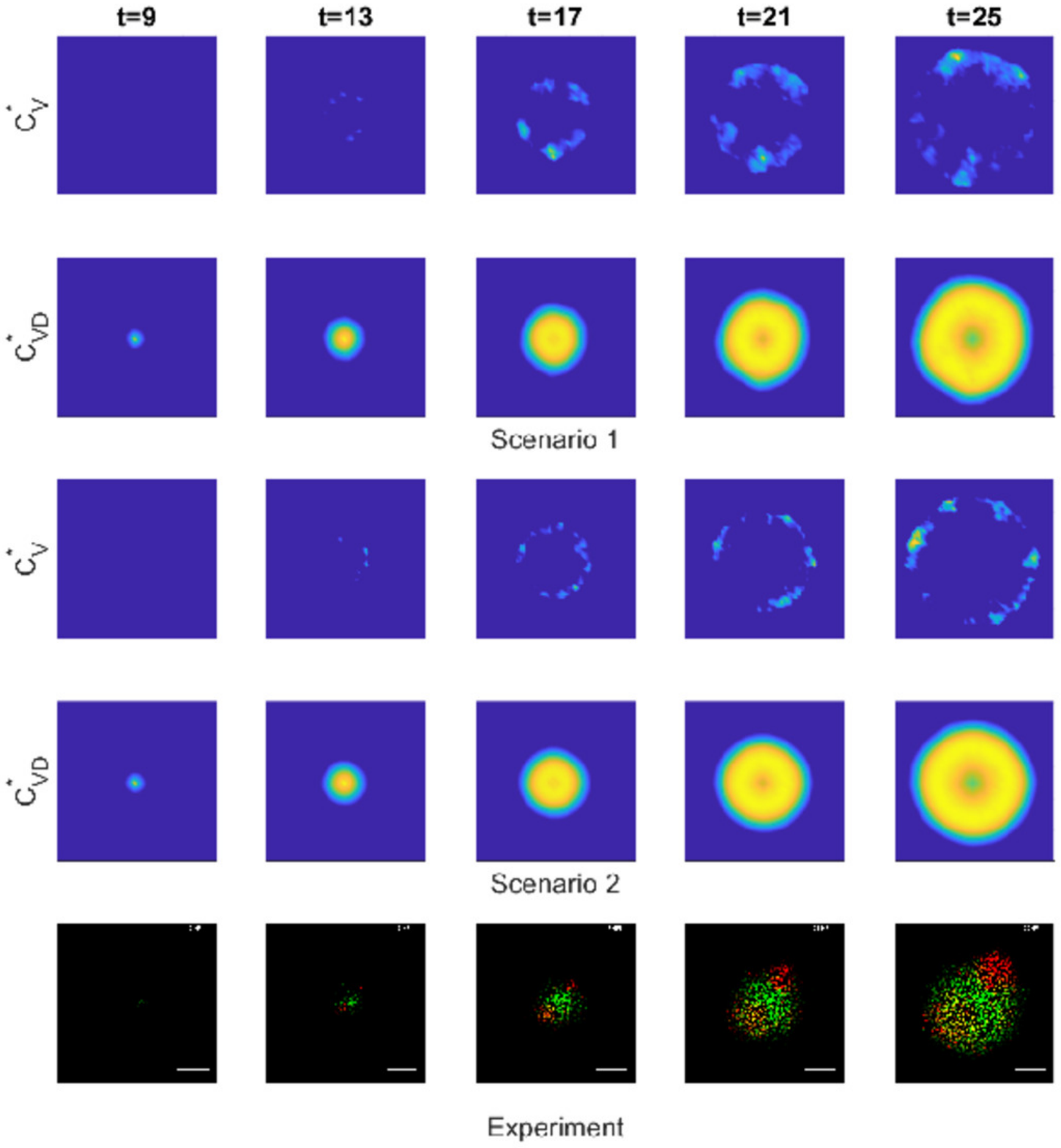}
\caption{{\bf Dynamics of virus and DIP in cells in 2 scenarios simulations with $C_{VD22,22}(0)=200$ and representative experiment with initial DIP equal to 84.}
Row 1, 2 are time series plots of virus in cells ($C^*_V$) and DIP in cells ($C^*_{VD}$) growth in Scenario 1 (infected cells produce virus and DIPs through cell bursting); 
Row 3, 4 are time series plots of those in Scenario 2 (infected cells keep producing viruses and DIPs); Row 5 is the representative experimental results.}
\label{case4}
\end{figure}

\begin{figure}[H]
\includegraphics[width=0.85\textwidth]{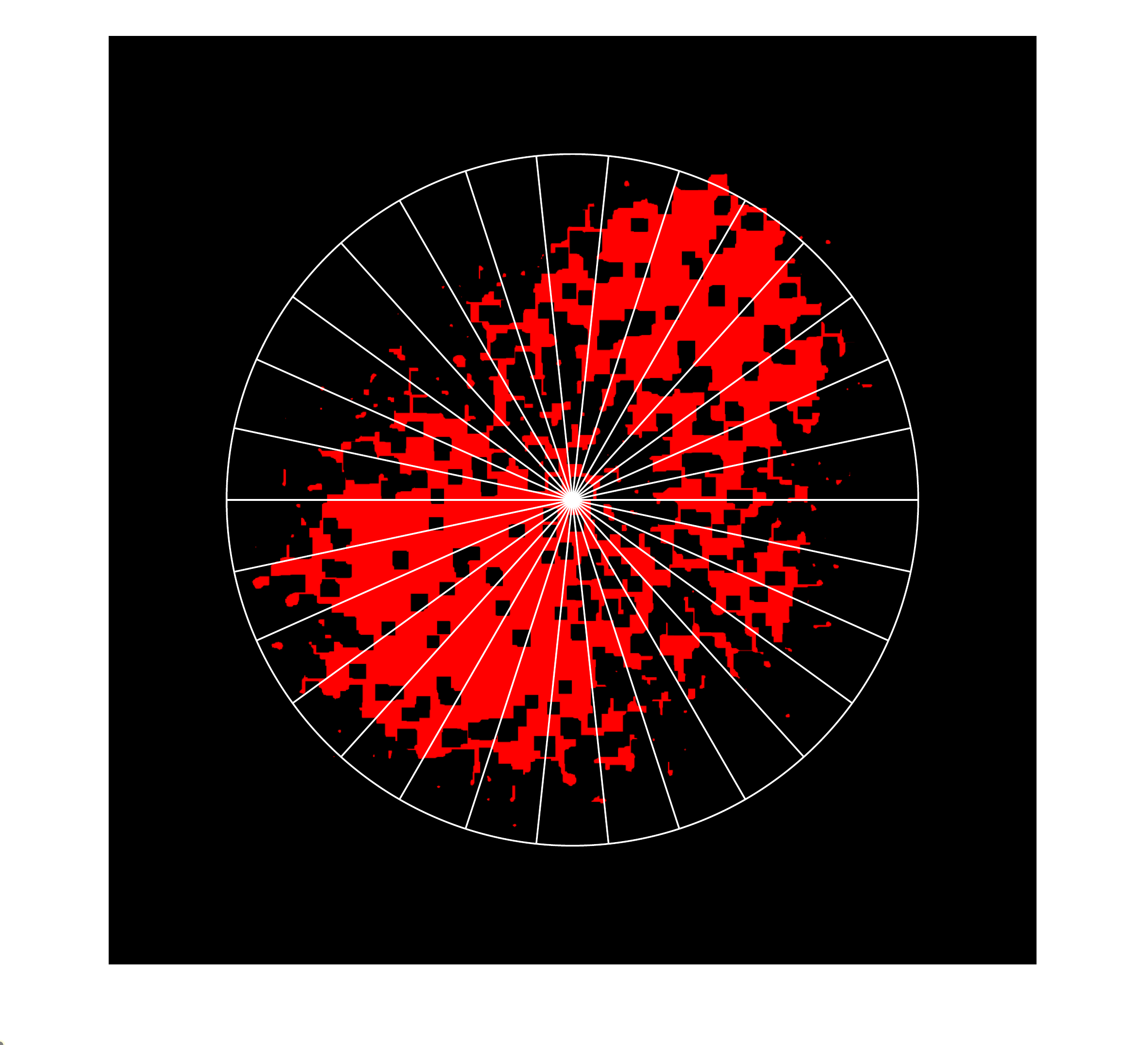}
\caption{{\bf Strata used in the computation of $q$-statistics.}}
\label{stratas}
\end{figure}

\begin{figure}[H]
\includegraphics[width=0.85\textwidth]{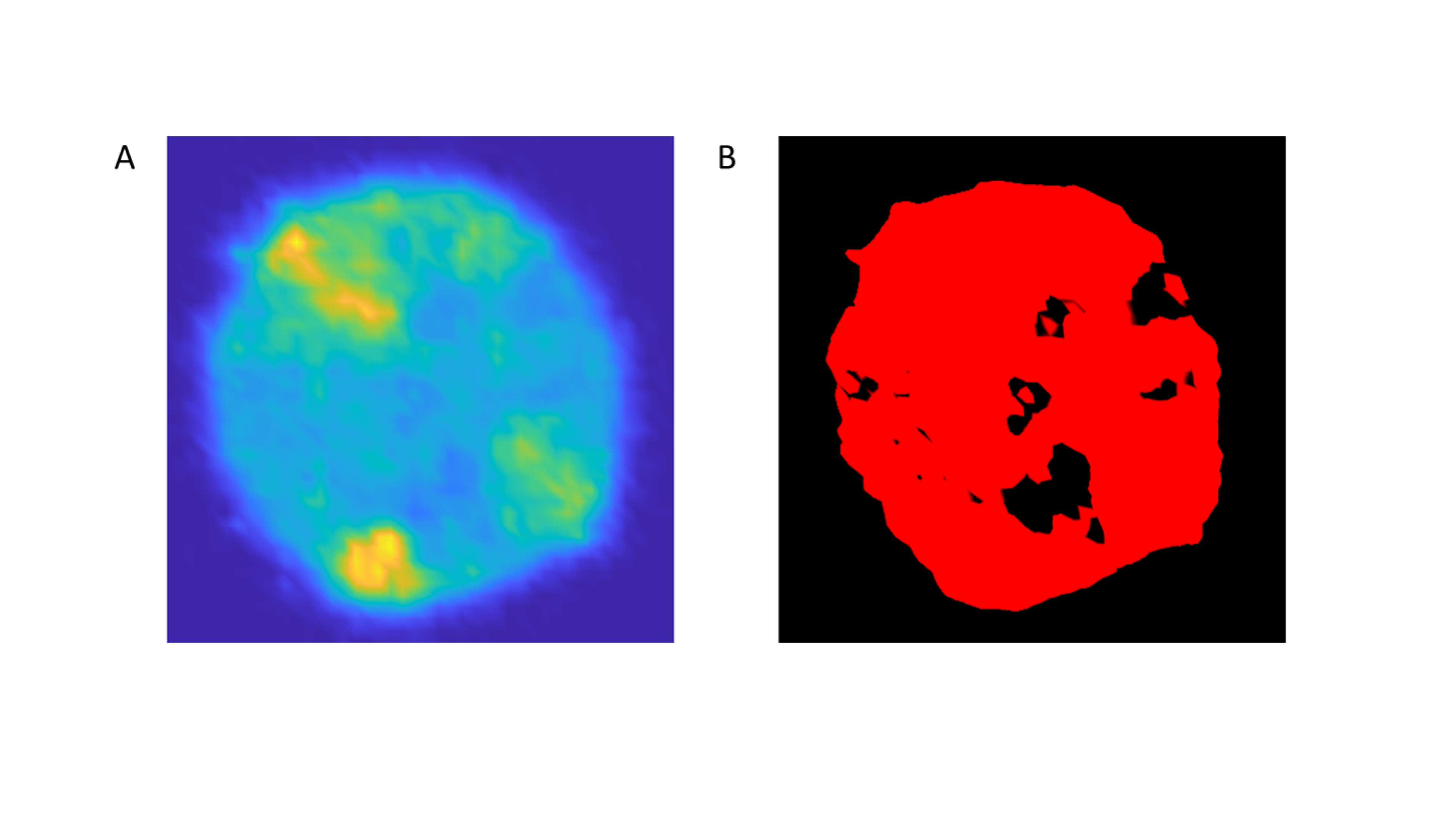}
\caption{{\bf A representative example to illustrate the image processing of simulation results for computing the $q$-statistic.}
A: the spatial distribution of $C_V^*$ in 2D with the same color bar as Fig~\ref{case1}-\ref{case3}, which is then converted to binary. B: all points where $C^*_V<50$ are black ((0,0,0) in RGB); otherwise, they are red ((255,0,0) in RGB).}
\label{simu_binary}
\end{figure}

\section*{Acknowledgments}
This work was supported by RGC General Research Fund (Project No. CityU 11301821).
The authors are grateful to Dr. John Yin for helpful discussions and for sharing the data in \cite{Baltes2017}.

%
%
%

\end{document}